\documentclass[11pt]{article}

\usepackage[margin=0.85in]{geometry} 
\usepackage{graphicx} 
\usepackage{amsfonts}
\usepackage[round]{natbib}
\usepackage{amssymb,amsmath,amsthm}
\usepackage{bm}
\usepackage{algorithm,algpseudocode} 
\algnewcommand{\Inputs}[1]{%
  \State \textbf{Inputs:}
  \Statex \hspace*{\algorithmicindent}\parbox[t]{.8\linewidth}{\raggedright #1}
}
\algnewcommand{\Initialize}[1]{%
  \State \textbf{Initialize:}
  \Statex \hspace*{\algorithmicindent}\parbox[t]{.8\linewidth}{\raggedright #1}
}

\usepackage{todonotes}

\usepackage{xcolor}

\newcommand{\prog}[1]{\textsf{#1}}

\newcommand{\code}[1]{\texttt{#1}}

\usepackage{enumitem}
\usepackage[colorlinks=true, linkcolor=blue, citecolor=blue, urlcolor=blue]{hyperref}
\usepackage{ctable}

\allowdisplaybreaks

\newtheorem{proposition}{Proposition}

\def\bb{{\boldsymbol{b}}}

\def\bx{{\boldsymbol{x}}}
\def\by{{\boldsymbol{y}}}

\def\bz{{\boldsymbol{z}}}

\def\bB{{\boldsymbol{B}}}

\def\bX{{\boldsymbol{X}}}
\def\bY{{\boldsymbol{Y}}}

\def\bZ{{\boldsymbol{Z}}}

\def\bbeta{\boldsymbol{\beta}}

\def\btheta{\boldsymbol{\theta}}

\def\bmu{\boldsymbol{\mu}}

\def\bpi{\boldsymbol{\pi}}

\def\bvartheta{\boldsymbol{\vartheta}}

\def\bDelta{\boldsymbol{\Delta}}

\def\bSigma{\boldsymbol{\Sigma}}

\def\bOmega{\boldsymbol{\Omega}}

\title{
A New Look at Gaussian Mixtures in the Presence of Missing-at-Random Responses and Covariates
}
\author{Hung Tong$^1$, Antonio Punzo$^2$, Cristina Tortora$^3$}
\date{$^1$ Rowan University, USA \quad $^2$ University of Catania, Italy \quad $^3$ San Jos\'e State University, USA}

\begin{document}

\maketitle

\begin{abstract}
Missing values present a common challenge in statistical modeling, so handling them properly is an important research direction. Among the various mechanisms that can generate missing values, the most common is the missing-at-random (MAR) mechanism, where the probability of missingness depends only on observed and not on unobserved data.
This paper addresses the problem of estimating a multivariate linear regression model with multiple random covariates in the presence of MAR values in both the response and covariate spaces, using a maximum likelihood (ML) framework.
The proposed methodology models the joint distribution of responses and covariates through a conditional–marginal factorization of a multivariate Gaussian distribution, which can be interpreted as a reparameterization of the multivariate normal distribution when the variables can be naturally partitioned into responses and covariates.
Parameter estimation is performed via the expectation–maximization (EM) algorithm, which facilitates the imputation of missing values while preserving the distinct roles of responses and covariates.
We extend this framework to the model-based clustering setting by considering a mixture of multivariate linear regressions with multiple random covariates. This extension enables soft clustering under incomplete data and accommodates MAR values in both the multivariate responses and covariates. Hence, it represents one of the most general model-based clustering solutions for regression data available in the current literature.
The effectiveness of the methodology is demonstrated through a simulation study, and the advantages of the proposed reparameterization are illustrated on the Automobile dataset, which contains missing values.
\end{abstract}

\vspace{1em}

\section{Introduction}
\label{sec:Introduction}

Missing values are commonly encountered in real-world data analyses.  They are particularly prevalent in psychology \citep{allison2009missing}, medical studies \citep{carpenter2021missing}, and survey research \citep{brick1996handling, mirzaei2022missing}, but they can occur in any application, and their treatment presents a significant challenge.
When the probability of missingness depends on observed data but remains independent of unobserved data, the missing-data mechanism is referred to as missing at random \citep[MAR;][Section 1.3]{little2019statistical}.
Addressing MAR values effectively is crucial for ensuring the validity and reliability of statistical procedures.

In this paper, we adopt the MAR assumption and consider incomplete datasets in which each row and each column contains at least one observed value---a condition representing the most general unconstrained scenario. Furthermore, we study the linear regression model in its most general multivariate and multiple form, where the $d$ available real-valued variables can be reasonably partitioned into $d_{\bY}$ responses $\bY$ and $d_{\bX}$ covariates $\bX$, with $d=d_{\bX}+d_{\bY}$. 
This partitioning is based on either the intrinsic meaning of the variables or the researcher's specific analytical interests. Therefore, it is important to preserve this partitioning to ensure the effectiveness of the analysis.
In this context, handling MAR values becomes more complex because missing values may occur in $\bY$, $\bX$, or both. 

For regression analysis, least-squares methods and maximum likelihood estimation are common approaches to address MAR values in $\bY$, $\bX$, or both \citep{little1992regression}. However, least squares methods are valid only under the missing completely at random (MCAR) mechanism---a special case of the MAR mechanism that is often unrealistic in practice. In contrast, maximum likelihood (ML) estimation offers several advantages under the MAR mechanism. Specifically, when data are missing at random, the missing data mechanism becomes ignorable \citep[][Section 6.2]{little2019statistical}. ML estimation also yields estimates that are consistent and asymptotically efficient \citep{little1992regression, schafer1997analysis}. Furthermore, it utilizes all available information, thereby reducing potential biases and improving efficiency compared to complete-case analysis \citep{rubin1976inference}. For these reasons, we adopt ML estimation using the expectation-maximization (EM) algorithm (\citealp{orchard1972missing}, \citealp{DemLai77}, and \citealp{little2019statistical}) for our proposed methodology. As thoroughly documented in \citet{little1992regression}, ML estimation for a general missing data pattern necessitates iterative methods. Accordingly, the EM algorithm offers a useful tool, as it is specifically designed to handle such estimation processes in the presence of missing values.
Building on this foundation, the EM algorithm has, in fact, been extensively employed in more recent literature to address datasets with MAR mechanisms in the context of cluster analysis, where no distinction is made between responses and covariates—the goal being to identify homogeneous groups of observations. Within this line of work, \citet{ghahramani_learning_1994} investigated the normal mixture, \citet{wang_robust_2004} the $t$ mixture, \citet{tong_model-based_2022} the contaminated normal mixture, \citet{wei_mixtures_2019} the mixtures of generalized hyperbolic and skew-$t$ distributions, and \citet{PillaySN} the mixture of scale-mixtures of skew-normal distributions, to cite a few. The results of these recent developments have already been applied in fields such as medicine and biology \citep{vellido2005handling, scheres2009introducing}.

In our proposed methodology, we employ a multivariate linear regression model with random covariates, in which the joint distribution of the responses $\bY$ and covariates $\bX$ is factorized into the product of the marginal distribution of $\bX$ and the conditional distribution of $\bY \mid \bX = \bx$. Additionally, both $\bX$ and $\bY \mid \bX = \bx$ are assumed to follow multivariate Gaussian distributions. As described in later sections, the normality assumption lays an important theoretical foundation for handling MAR values in model fitting. Although the factorization corresponds to modeling $\bY$ and $\bX$ jointly as a $d$-dimensional Gaussian distribution, it aligns with the distinct roles of the variables as responses and covariates, thus enhancing parameter interpretability. 
Beyond this, reparameterizing the joint Gaussian distribution through the conditional–marginal factorization provides a novel perspective on the implicit model-based imputation carried out during the EM iterations, elucidating how observed covariates contribute to the conditional expectations of missing responses and vice versa.
In addition to enabling MAR mechanisms within a regression framework with random covariates, a further innovation arises when this framework is extended to the cluster analysis context. In practice, incompleteness due to MAR values often coexists with unobserved heterogeneity induced by latent clusters, resulting in additional missing information regarding cluster membership. The proposed approach naturally accommodates this scenario under a model-based clustering framework by adopting a mixture of linear regressions with random covariates, assuming Gaussian distributions for both $\bY|\bX=\bx$ and $\bX$ within each component \citep{dang2017multivariate}. Unlike most existing mixture-of-regressions models (with fixed covariates), which restrict missingness to $\bY$, our formulation allows MAR values to occur in $\bY$, $\bX$, or both—an aspect seldom addressed in the current literature.

Since this formulation is a reparameterization of the joint $d$-variate Gaussian distribution (and, in the clustering extension, of a finite mixture of $d$-variate Gaussians), all classical asymptotic properties of ML estimation are automatically inherited. 
Under standard regularity conditions \citep{wald1949, lehmanncasella}, the ML estimator is consistent, asymptotically normal, and asymptotically efficient in Gaussian models, while for Gaussian mixtures, the established theory \citep{redner1984, mclachlan2000} ensures local consistency and characterizes the limiting distribution of estimators. 
The EM algorithm also preserves its monotonicity and convergence to stationary points \citep{wu1983}. 
Thus, although conceptually simple, the proposed parametrization provides full theoretical soundness while offering the interpretative advantages discussed above.

Within this framework, the EM algorithm facilitates ML estimation in the presence of MAR data, extending the advantages of the non-clustering case to a soft-clustering paradigm where each observation is assigned a probabilistic cluster membership. Hence, the proposed methodology unifies the treatment of MAR data and latent heterogeneity within a coherent ML framework that preserves interpretability and enhances robustness.

In this context, by treating covariates as random variables, they actively participate in the clustering process, thus influencing cluster formation. This property, referred to in the literature as \textit{assignment dependence}, stands in contrast to the \textit{assignment independence} assumption inherent in fixed-covariate approaches, which can be a significant limitation in many real-world data applications \citep{hennig2000identifiablity}. The need for simultaneous identification of homogeneous groups and clusterwise regression arises in many applied fields, including welding processes \citep{ganjigatti2007global, lee2019global}, transportation \citep{luo2006pavement, veeramisti2021clusterwise}, medicine \citep{miller2012clustering, kato2020clustering, zamaninasab2022cluster}, and environmental sciences \citep{bagirov2017prediction}, among others. Incorporating random covariates and MAR treatment further broadens the applicability of such methods.

To illustrate our model-based clustering method in the presence of MAR values, we consider the Automobile dataset \citep{schlimmer1985automobile} available from the UCI Machine Learning Repository. 
This dataset includes various automobile characteristics, such as risk ratings, normalized losses, and prices, with missing values in several attributes. 
Our analysis highlights the effectiveness of our approach in handling missing data while providing insightful interpretations and parameter estimates that incorporate the variability due to imputation.

The rest of the article is organized as follows. 
Section~\ref{sec:backgraund} contains some background information on multivariate linear regression with multiple random covariates, including some necessary propositions. Section~\ref{sec:MLR mar} introduces multivariate linear regression with multiple random covariates with missing values at random. 
Cluster-wise multivariate linear regression with multiple random covariates and missing values is introduced in Section \ref{sec:clusterwise_missing}. 
In Section~\ref{sec:Simulation_Study}, we describe the results of studies on simulated data to evaluate the performance of our model-based clustering approach. 
The assessment focuses on parameter recovery and clustering accuracy across varying levels of missingness in both responses and covariates, as well as the method’s robustness to deviations from Gaussian assumptions.
Section~\ref{sec:real data} illustrates the application to the Automobile data set. 
Finally, Section~\ref{sec:conclusion} concludes the paper.

\section{Multivariate linear regression with multiple random covariates}
\label{sec:backgraund}

This section provides an overview of the paper's essential background. Specifically, we briefly introduce multivariate linear regression with random covariates. 
We also include some propositions, with the corresponding proofs, that will be used in subsequent sections.
All vectors in this paper are assumed to be column vectors.

Let $p (\bx, \by; \bvartheta)$ be the parametric model, parameterized by $\bvartheta$, considered for the joint density of $\left(\bX^\top, \bY^\top\right)^\top$.
Under a multivariate linear regression with multiple random covariates, the joint density can be written as
\begin{align}\label{eq:densreg}
    p (\bx, \by; \bvartheta) = p (\by \mid \bx; \bvartheta_\bY) \, p (\bx; \bvartheta_\bX),
\end{align}
where $p (\by \mid \bx; \bvartheta_\bY)$ is the conditional density of $\bY|\bX=\bx$ with parameters $\bvartheta_\bY$ and $p (\bx; \bvartheta_\bX)$ is the marginal density of $\bX$ with parameters $\bvartheta_\bX$. 
We assume multivariate Gaussian distributions for both marginal and conditional densities such that
\begin{align*}
    \bX \sim \mathcal{G}_{d_\bX} \left( \bmu_\bX, \bSigma_\bX \right) \quad \text{and} \quad \bY \mid \bX=\bx \sim \mathcal{G}_{d_\bY} \left( \bmu_\bY (\bx; \bbeta), \bSigma_\bY \right),
\end{align*}
that is, $\bX$ follows a $d_\bX$-dimensional Gaussian distribution with mean $\bmu_\bX$ and covariance $\bSigma_\bX$, while $\bY|\bX=\bx$ follows a $d_\bY$-dimensional Gaussian distribution with covariance $\bSigma_\bY$ and mean, under a multivariate linearity assumption, given by
\begin{align*}
    \underset{d_\bY \times 1}{\bmu_\bY (\bx; \bbeta)} &= \bbeta^\top \begin{bmatrix} 1 \\ \bx \end{bmatrix} =  \bb_0 + \bB \bx.
\end{align*}
In the above, $\bbeta$ is a $(1 + d_\bX) \times d_\bY$ matrix of regression coefficients such that
\begin{align*}
    \underset{d_\bY \times (1 + d_\bX)}{\bbeta^\top} = \begin{bmatrix}
        \underset{d_\bY \times 1}{\bb_{0}}  & \underset{d_\bY \times d_\bX}{\bB}
    \end{bmatrix},
\end{align*}
where
\begin{align*}
    \underset{d_\bY \times 1}{\bb_{0}} = \begin{bmatrix} b_{01} \\ b_{02} \\ \vdots \\ b_{0d_\bY} \end{bmatrix} \quad \text{and} \quad \underset{d_\bY \times d_\bX}{\bB} = \begin{bmatrix} 
b_{11} & b_{21} & \cdots & b_{ d_\bX 1} \\
b_{12} & b_{22} & \cdots & b_{ d_\bX 2} \\
\vdots & \vdots  & \ddots & \vdots \\
b_{1d_\bY} & b_{2d_\bY} & \cdots & b_{ d_\bX d_\bY}
\end{bmatrix} 
\end{align*}
are the vector of intercepts and the matrix of slopes, respectively. The joint density in \eqref{eq:densreg} becomes
\begin{align}
    p (\bx, \by; \bvartheta) &= \phi \left( \by; \bmu_\bY (\bx; \bbeta), \bSigma_\bY \right) \, \phi \left( \bx; \bmu_\bX, \bSigma_\bX \right),\label{eq:jointreg}
\end{align}
where $\phi (\cdot; \bmu, \bSigma)$ denotes the density of a multivariate Gaussian random vector with mean vector $\bmu$ and covariance matrix $\bSigma$. According to \citet[Ch.~2.3]{bishop2006pattern}, these distributional assumptions lead to
\begin{align*}
    \begin{bmatrix} \bX \\[1ex] \bY \end{bmatrix} \sim \mathcal{G}_{d_\bX + d_\bY} \left( \begin{bmatrix} \bmu_{\bX} \\[1ex] \tilde{\bmu}_{\bY} \end{bmatrix}, \begin{bmatrix} \bSigma_{\bX \bX} & \bSigma_{\bX \bY} \\[1ex] \bSigma_{\bY \bX} & \bSigma_{\bY \bY} \end{bmatrix} \right),
\end{align*}
where
\begin{align*}
    \underset{(d_\bY \times 1)}{\tilde{\bmu}_{\bY}} &= \bbeta^\top \begin{bmatrix} 1 \\ \bmu_\bX \end{bmatrix} = \bb_0 + \bB \bmu_{\bX}, \\[1ex]
    \underset{(d_\bX \times d_\bX)}{\bSigma_{\bX \bX}} &= \bSigma_{\bX}, \quad \underset{(d_\bX \times d_\bY)}{\bSigma_{\bX \bY}} = \bSigma_{\bX} \bB^\top, \quad \underset{(d_\bY \times d_\bX)}{\bSigma_{\bY \bX}} = \bB \bSigma_{\bX}, \quad \text{and} \quad \underset{(d_\bY \times d_\bY)}{\bSigma_{\bY \bY}} = \bB \bSigma_{\bX} \bB^\top + \bSigma_{\bY}.
\end{align*}



Now, we further suppose that $\bX$ is partitioned as $\bX = \left( \bX^\top_1, \bX^\top_2 \right)^\top$, where $\bX_1$ and $\bX_2$ are sub-vectors with dimensions $d_{\bX_1}$ and $d_{\bX_2}$, respectively. Moreover, $\bY$ is partitioned as $\bY = \left( \bY^\top_1, \bY^\top_2 \right)^\top$, where $\bY_1$ and $\bY_2$ are sub-vectors with dimensions $d_{\bY_1}$ and $d_{\bY_2}$, respectively. Then, we have the following decompositions
of the mean vectors

\begin{align*}
    \bmu_\bX &= \begin{bmatrix} \underset{(d_{\bX_1} \times 1)}{\bmu_{\bX_1}} \\[3ex] \underset{(d_{\bX_2} \times 1)}{\bmu_{\bX_2}} \end{bmatrix}, \qquad \tilde{\bmu}_{\bY} = \begin{bmatrix} \underset{(d_{\bY_1} \times 1)}{\tilde{\bmu}_{\bY_1}} \\[3ex] \underset{(d_{\bY_2} \times 1)}{\tilde{\bmu}_{\bY_2}} \end{bmatrix}
\end{align*}
and the following decompositions of the covariance matrices
\begin{align*}
    \bSigma_{\bX \bX} &= \begin{bmatrix} \underset{(d_{\bX_1} \times d_{\bX_1})}{\bSigma_{\bX_1 \bX_1}} & \underset{(d_{\bX_1} \times d_{\bX_2})}{\bSigma_{\bX_1 \bX_2}} \\[4ex] \underset{(d_{\bX_2} \times d_{\bX_1})}{\bSigma_{\bX_2 \bX_1}} & \underset{(d_{\bX_2} \times d_{\bX_2})}{\bSigma_{\bX_2 \bX_2}} \end{bmatrix}, & \bSigma_{\bX \bY} &= \begin{bmatrix} \underset{(d_{\bX_1} \times d_{\bY_1})}{\bSigma_{\bX_1 \bY_1}} & \underset{(d_{\bX_1} \times d_{\bY_2})}{\bSigma_{\bX_1 \bY_2}} \\[4ex] \underset{(d_{\bX_2} \times d_{\bY_1})}{\bSigma_{\bX_2 \bY_1}} & \underset{(d_{\bX_2} \times d_{\bY_2})}{\bSigma_{\bX_2 \bY_2}} \end{bmatrix}, \\[2ex]
    \bSigma_{\bY \bX} &= \begin{bmatrix} \underset{(d_{\bY_1} \times d_{\bX_1})}{\bSigma_{\bY_1 \bX_1}} & \underset{(d_{\bY_1} \times d_{\bX_2})}{\bSigma_{\bY_1 \bX_2}} \\[4ex] \underset{(d_{\bY_2} \times d_{\bX_1})}{\bSigma_{\bY_2 \bX_1}} & \underset{(d_{\bY_2} \times d_{\bX_2})}{\bSigma_{\bY_2 \bX_2}} \end{bmatrix}, & \bSigma_{\bY \bY} &= \begin{bmatrix} \underset{(d_{\bY_1} \times d_{\bY_1})}{\bSigma_{\bY_1 \bY_1}} & \underset{(d_{\bY_1} \times d_{\bY_2})}{\bSigma_{\bY_1 \bY_2}} \\[4ex] \underset{(d_{\bY_2} \times d_{\bY_1})}{\bSigma_{\bY_2 \bY_1}} & \underset{(d_{\bY_2} \times d_{\bY_2})}{\bSigma_{\bY_2 \bY_2}} \end{bmatrix}.
\end{align*}
Again, following \citet[Ch.~2.3]{bishop2006pattern}, based on the properties of the multivariate Gaussian distribution, we can write the overall joint distribution as
\begin{align}
\begin{bmatrix} \bX_1 \\ \bX_2 \\ \bY_1 \\ \bY_2 \end{bmatrix}  \sim  \mathcal{G}_{d_\bX + d_\bY} \left( \begin{bmatrix} \bmu_{\bX_1} \\ \bmu_{\bX_2} \\ \tilde{\bmu}_{\bY_1} \\ \tilde{\bmu}_{\bY_2} \end{bmatrix}, \begin{bmatrix}
\bSigma_{\bX_1 \bX_1} & \bSigma_{\bX_1 \bX_2} & \bSigma_{\bX_1 \bY_1} & \bSigma_{\bX_1 \bY_2} \\
\bSigma_{\bX_2 \bX_1} & \bSigma_{\bX_2 \bX_2} & \bSigma_{\bX_2 \bY_1} & \bSigma_{\bX_2 \bY_2} \\
\bSigma_{\bY_1 \bX_1} & \bSigma_{\bY_1 \bX_2} & \bSigma_{\bY_1 \bY_1} & \bSigma_{\bY_1 \bY_2} \\
\bSigma_{\bY_2 \bX_1} & \bSigma_{\bY_2 \bX_2} & \bSigma_{\bY_2 \bY_1} & \bSigma_{\bY_2 \bY_2} \\
\end{bmatrix} \right). \label{eq:joint_X1_X2_Y1_Y2}
\end{align}

This definition is necessary to derive some conditional and joint distributions that will be used in the following sections. 
\begin{proposition} \label{prop:joint_X1_Y1}
The joint distribution of $\bX_1$ and $\bY_1$ is a multivariate Gaussian distribution as follows
\begin{align}
    \begin{bmatrix} \bX_1 \\ \bY_1 \end{bmatrix}  \sim  \mathcal{G}_{d_{\bX_1} + d_{\bY_1}} \left( \begin{bmatrix} \bmu_{\bX_1} \\ \tilde{\bmu}_{\bY_1} \end{bmatrix}, \begin{bmatrix}
\bSigma_{\bX_1 \bX_1} & \bSigma_{\bX_1 \bY_1}  \\
\bSigma_{\bY_1 \bX_1} & \bSigma_{\bY_1 \bY_1}
\end{bmatrix} \right). \label{eq:joint_X1_Y1}
\end{align}

\end{proposition}


\begin{proof}
According to~\eqref{eq:joint_X1_X2_Y1_Y2}, $\bX_1, \bX_2, \bY_1,$ and $\bY_2$ are jointly distributed as a multivariate Gaussian distribution. From Result 4.4 of \citet{johnson_applied_2007}, their subset formed by $\bX_1$ and $\bY_1$ also follows a multivariate Gaussian distribution with the mean vector and covariance matrix to be
\begin{align*}
    \begin{bmatrix} \bmu_{\bX_1} \\ \tilde{\bmu}_{\bY_1} \end{bmatrix}\quad \text{and}\quad \begin{bmatrix}
\bSigma_{\bX_1 \bX_1} & \bSigma_{\bX_1 \bY_1}  \\
\bSigma_{\bY_1 \bX_1} & \bSigma_{\bY_1 \bY_1}
\end{bmatrix}.
\end{align*}

\end{proof}


\begin{proposition} \label{prop:joint_X2_X1_Y1}
The joint distribution of $\bX_2, \bX_1$, and $\bY_1$ is a multivariate Gaussian distribution as follows
\begin{align}
    \begin{bmatrix} \bX_2 \\ \bX_1 \\ \bY_1 \end{bmatrix}  \sim \mathcal{G}_{d_{\bX_2} + d_{\bX_1} + d_{\bY_1}} \left( \begin{bmatrix} \bmu_{\bX_2} \\ \bmu_{\bX_1} \\ \tilde{\bmu}_{\bY_1} \end{bmatrix}, \begin{bmatrix}
\bSigma_{\bX_2 \bX_2} & \bSigma_{\bX_2 \bX_1} & \bSigma_{\bX_2 \bY_1} \\
\bSigma_{\bX_1 \bX_2} & \bSigma_{\bX_1 \bX_1} & \bSigma_{\bX_1 \bY_1} \\
\bSigma_{\bY_1 \bX_2} & \bSigma_{\bY_1 \bX_1} & \bSigma_{\bY_1 \bY_1}
\end{bmatrix} \right). \label{eq:joint_X2_X1_Y1}
\end{align}
\end{proposition}


\begin{proof}
The result can be shown using the same reasoning as in Proposition~\ref{prop:joint_X1_Y1}.
\end{proof}


\begin{proposition} \label{prop:cond_X2_x1_y1}
The conditional distribution of $\bX_2$ given $\bX_1 = \bx_1$ and $\bY_1 = \by_1$ is a multivariate Gaussian distribution as follows
\begin{align}
    \bX_2 \mid \bX_1 = \bx_1, \bY_1 = \by_1 \sim  \mathcal{G}_{d_{\bX_2}} \left( \bmu_{\bX_2 \mid \bx_1, \by_1}, \bSigma_{\bX_2 \mid \bx_1, \by_1} \right), \label{eq:cond_X2_x1_y1}
\end{align}
where
\begin{align*}
    \bmu_{\bX_2 \mid \bx_1, \by_1} &= \bmu_{\bX_2} + \begin{bmatrix} \bSigma_{\bX_2 \bX_1} & \bSigma_{\bX_2 \bY_1} \end{bmatrix} \begin{bmatrix} \bSigma_{\bX_1 \bX_1} & \bSigma_{\bX_1 \bY_1} \\ \bSigma_{\bY_1 \bX_1} & \bSigma_{\bY_1 \bY_1} \end{bmatrix}^{-1} \left( \begin{bmatrix} \bx_1 \\ \by_1 \end{bmatrix} - \begin{bmatrix} \bmu_{\bX_1} \\ \tilde{\bmu}_{\bY_1} \end{bmatrix} \right) \\[1ex]
    \text{and} \quad \bSigma_{\bX_2 \mid \bx_1, \by_1} &= \bSigma_{\bX_2 \bX_2} - \begin{bmatrix} \bSigma_{\bX_2 \bX_1} & \bSigma_{\bX_2 \bY_1} \end{bmatrix} \begin{bmatrix} \bSigma_{\bX_1 \bX_1} & \bSigma_{\bX_1 \bY_1} \\ \bSigma_{\bY_1 \bX_1} & \bSigma_{\bY_1 \bY_1} \end{bmatrix}^{-1} \begin{bmatrix} \bSigma_{\bX_1 \bX_2} \\ \bSigma_{\bY_1 \bX_2} \end{bmatrix}.
\end{align*}
\end{proposition}


\begin{proof}
The result follows by applying Proposition~\ref{prop:joint_X2_X1_Y1} and Result 4.6 of \citet{johnson_applied_2007} regarding the conditional distribution of a multivariate Gaussian distribution. 
The mean vector and covariance matrix are determined using the formulas given in Result 4.6 of \citet{johnson_applied_2007}.
\end{proof}


\begin{proposition} \label{prop:joint_Y2_X1_Y1}
The joint distribution of $\bY_2, \bX_1$, and $\bY_1$ is a multivariate Gaussian distribution as follows
\begin{align}
    \begin{bmatrix} 
    \bY_2 \\ 
    \bX_1 \\ 
    \bY_1 
    \end{bmatrix} \sim  \mathcal{G}_{d_{\bY_2} + d_{\bX_1} + d_{\bY_1}} \left( \begin{bmatrix} \tilde{\bmu}_{\bY_2} \\ \bmu_{\bX_1} \\ \bmu_{\bY_1} \end{bmatrix}, \begin{bmatrix}
\bSigma_{\bY_2 \bX_2} & \bSigma_{\bY_2 \bX_1} & \bSigma_{\bY_2 \bY_1} \\
\bSigma_{\bX_1 \bX_2} & \bSigma_{\bX_1 \bX_1} & \bSigma_{\bX_1 \bY_1} \\
\bSigma_{\bY_1 \bX_2} & \bSigma_{\bY_1 \bX_1} & \bSigma_{\bY_1 \bY_1}
\end{bmatrix} \right). \label{eq:joint_Y2_X1_Y1}
\end{align}
    
\end{proposition}


\begin{proof}
The result can be shown using the same reasoning as in Proposition~\ref{prop:joint_X1_Y1}.
\end{proof}


\begin{proposition}\label{prop:cond_Y2_X1_Y1}
The conditional distribution of $\bY_2$ given $\bX_1 = \bx_1$ and $\bY_1 = \by_1$ is a multivariate Gaussian distribution as follows
\begin{align*}
    \bY_2 \mid \bX_1 = \bx_1, \bY_1 = \by_1 \sim  \mathcal{G}_{d_{\bY_2}} \left( \bmu_{\bY_2 \mid \bx_1, \by_1}, \bSigma_{\bY_2 \mid \bx_1, \by_1} \right),
\end{align*}
where
\begin{align*}
    \bmu_{\bY_2 \mid \bx_1, \by_1} &= \tilde{\bmu}_{\bY_2} + \begin{bmatrix} \bSigma_{\bY_2 \bX_1} & \bSigma_{\bY_2 \bY_1} \end{bmatrix} \begin{bmatrix} \bSigma_{\bX_1 \bX_1} & \bSigma_{\bX_1 \bY_1} \\ \bSigma_{\bY_1 \bX_1} & \bSigma_{\bY_1 \bY_1} \end{bmatrix}^{-1} \left( \begin{bmatrix} \bx_1 \\ \by_1 \end{bmatrix} - \begin{bmatrix} \bmu_{\bX_1} \\ \tilde{\bmu}_{\bY_1} \end{bmatrix} \right) \\[1ex]
    \text{and} \quad \bSigma_{\bY_2 \mid \bx_1, \by_1} &= \bSigma_{\bY_2 \bY_2} - \begin{bmatrix} \bSigma_{\bY_2 \bX_1} & \bSigma_{\bY_2 \bY_1} \end{bmatrix} \begin{bmatrix} \bSigma_{\bX_1 \bX_1} & \bSigma_{\bX_1 \bY_1} \\ \bSigma_{\bY_1 \bX_1} & \bSigma_{\bY_1 \bY_1} \end{bmatrix}^{-1} \begin{bmatrix} \bSigma_{\bX_1 \bY_2} \\ \bSigma_{\bY_1 \bY_2} \end{bmatrix}.
\end{align*}

\end{proposition}


\begin{proof}
The result follows by applying Proposition~\ref{prop:joint_Y2_X1_Y1} and Result 4.6 of \citet{johnson_applied_2007} regarding the conditional distribution of a multivariate Gaussian distribution. The mean vector and covariance matrix are determined using the formulas given in Result 4.6.
\end{proof}


\begin{proposition}\label{prop:cond_Y2_X2_X1_Y1}
The conditional distribution of the random vector $(\bY_2, \bX_2)$ given $\bX_1 = \bx_1$ and $\bY_1 = \by_1$ is a multivariate Gaussian distribution as follows
\begin{align*}
    \bY_2, \bX_2 \mid \bX_1 = \bx_1, \bY_1 = \by_1 \sim  \mathcal{G}_{d_{\bY_2} + d_{\bX_2}} \left( \bmu_{\bY_2, \bX_2 \mid \bx_1, \by_1}, \bSigma_{\bY_2, \bX_2 \mid \bx_1, \by_1} \right),
\end{align*}
where
\begin{align*}
    \bmu_{\bY_2, \bX_2 \mid \bx_1, \by_1} &= \begin{bmatrix} \tilde{\bmu}_{\bY_2} \\ \bmu_{\bX_2} \end{bmatrix} + \begin{bmatrix} \bSigma_{\bY_2 \bY_1} & \bSigma_{\bY_2 \bX_1} \\ \bSigma_{\bX_2 \bY_1} & \bSigma_{\bX_2 \bX_1} \end{bmatrix} \begin{bmatrix} \bSigma_{\bY_1 \bY_1} & \bSigma_{\bY_1 \bX_1} \\ \bSigma_{\bX_1 \bY_1} & \bSigma_{\bX_1 \bX_1} \end{bmatrix}^{-1} \left( \begin{bmatrix} \bx_1 \\ \by_1 \end{bmatrix} - \begin{bmatrix} \bmu_{\bX_1} \\ \tilde{\bmu}_{\bY_1} \end{bmatrix} \right) \\[1ex]
    \text{and} \quad \bSigma_{\bY_2, \bX_2 \mid \bx_1, \by_1} &= \begin{bmatrix} \bSigma_{\bY_2 \bY_2} & \bSigma_{\bY_2 \bX_2} \\ \bSigma_{\bX_2 \bY_2} & \bSigma_{\bX_2 \bX_2} \end{bmatrix} - \begin{bmatrix} \bSigma_{\bY_2 \bY_1} & \bSigma_{\bY_2 \bX_1} \\
\bSigma_{\bX_2 \bY_1} & \bSigma_{\bX_2 \bX_1} \end{bmatrix} \begin{bmatrix} \bSigma_{\bY_1 \bY_1} & \bSigma_{\bY_1 \bX_1} \\ \bSigma_{\bX_1 \bY_1} & \bSigma_{\bX_1 \bX_1} \end{bmatrix}^{-1} \begin{bmatrix} \bSigma_{\bY_1 \bY_2} & \bSigma_{\bY_1 \bX_2} \\ \bSigma_{\bX_1 \bY_2} & \bSigma_{\bX_1 \bX_2} \end{bmatrix}.
\end{align*}
\end{proposition}


\begin{proof}
The result follows by applying~\eqref{eq:joint_X1_X2_Y1_Y2} and Result 4.6 of \citet{johnson_applied_2007} regarding the conditional distribution of a multivariate Gaussian distribution. The mean vector and covariance matrix are determined using the formulas given in Result 4.6.
\end{proof}

All the mentioned properties will be used to achieve ML parameter estimates. 
Given a random sample $\left(\bx_1^\top,\by_1^\top\right)^\top,\ldots,\left(\bx_n^\top,\by_n^\top\right)^\top$ from $\left(\bX^\top,\bY^\top\right)^\top$ with density in \eqref{eq:jointreg}, it is, therefore, important to define the likelihood function
\begin{align*}
    L (\bvartheta) &= \prod_{i = 1}^n \phi \left( \bx_i; \bmu_\bX, \bSigma_\bX \right) \phi \left( \by_i; \bmu_\bY \left( \bx_i; \bbeta \right), \bSigma_\bY \right).
\end{align*}
 We can further extend the formula using the density of a multivariate Gaussian distribution, i.e., 
\begin{align*}
    \phi \left( \bx_i; \bmu_\bX, \bSigma_\bX \right) = (2 \pi)^{-d_\bX / 2} \lvert \bSigma_\bX \rvert^{-1/2} \exp \left[ -\dfrac{1}{2} \left( \bx_i - \bmu_\bX \right)^\top \bSigma^{-1}_\bX \left( \bx_i - \bmu_\bX \right) \right]
\end{align*}
and
\begin{align*}
    \phi \left( \by_i; \bmu_\bY \left( \bx_i; \bbeta \right), \bSigma_\bY \right) &= (2 \pi)^{-d_\bY / 2} \lvert \bSigma_\bY \rvert^{-1/2} \exp \left[ -\dfrac{1}{2} \Big( \by_i - \bmu_\bY \left( \bx_i; \bbeta \right) \Big)^\top \bSigma^{-1}_\bY \Big( \by_i - \bmu_\bY \left( \bx_i; \bbeta \right) \Big) \right] \\[1ex]
    &= (2 \pi)^{-d_\bY / 2} \lvert \bSigma_\bY \rvert^{-1/2} \exp \left[ -\dfrac{1}{2} \left( \by_i - \bbeta^\top \begin{bmatrix} 1 \\ \bx_i \end{bmatrix} \right)^\top \bSigma^{-1}_\bY \left( \by_i - \bbeta^\top \begin{bmatrix} 1 \\ \bx_i \end{bmatrix} \right) \right].
\end{align*}
The corresponding log-likelihood can be written as follows 
\begin{align}
    l (\bvartheta) = l_{1} (\bvartheta_\bX) + l_{2}(\bvartheta_\bY), \label{eq:lc_regr}
\end{align}
where $\bvartheta_\bX=\{\bmu_\bX, \bSigma_\bX\}$, $\bvartheta_\bY=\{\bbeta, \bSigma_\bY\}$, and
\begin{align*}
    l_{1} (\bvartheta_\bX) &= -\dfrac{1}{2} \sum_{i = 1}^n \left[ d_\bX \log (2 \pi) + \log \lvert \bSigma_\bX \rvert + \left( \bx_i - \bmu_\bX \right)^\top \bSigma^{-1}_\bX \left( \bx_i - \bmu_\bX \right) \right] \\[1ex]
    l_{2} (\bvartheta_\bY) &= -\dfrac{1}{2} \sum_{i = 1}^n \left[ d_\bY \log (2 \pi) + \log \lvert \bSigma_\bY \rvert + \left( \by_i - \bmu_\bY \left( \bx_i; \bbeta \right) \right)^\top \bSigma^{-1}_\bY \left( \by_i - \bmu_\bY \left( \bx_i; \bbeta \right) \right) \right] \\
    &= -\dfrac{1}{2} \sum_{i = 1}^n \left[ d_\bY \log (2 \pi) + \log \lvert \bSigma_\bY \rvert + \left( \by_i - \bbeta^\top \begin{bmatrix} 1 \\ \bx_i \end{bmatrix} \right)^\top \bSigma^{-1}_\bY \left( \by_i - \bbeta^\top \begin{bmatrix} 1 \\ \bx_i \end{bmatrix} \right) \right].
\end{align*}
ML parameter estimates can be obtained using the EM algorithm detailed in the next section.

\section{Multivariate linear regression with multiple random covariates and missing values}\label{sec:MLR mar}

Recall that in this work, we focus on a missing at random (MAR) mechanism where missingness depends only on observed values of the data matrix \citep{little_statistical_2020}. 
In this context, each observation ($\bX_i, \bY_i$) can be decomposed into the observed and missing sub-vectors $\bX^o_i$ of dimension $d_{\bX_i}^o$, $\bX^m_i$ of dimension $d_{\bX_i}^m$, $\bY^o_i$ of dimension $d_{\bY_i}^{o}$, and $\bY^m_i$ of dimension $d_{\bY_i}^{m}$. Specifically,
\begin{align*}
    \bX_i = \begin{bmatrix} \underset{(d_{\bX_i}^o \times 1)}{\bX^o_i} \\[4ex] \underset{(d_{\bX_i}^m \times 1)}{\bX^m_i} \end{bmatrix} \quad \text{and} \quad \bY_i = \begin{bmatrix} \underset{(d_{\bY_i}^o \times 1)}{\bY^o_i} \\[4ex] \underset{(d_{\bY_i}^m \times 1)}{\bY^m_i} \end{bmatrix}.
\end{align*}
Note that the notations $o$ and $m$ do not imply the same pattern of missingness for all observations; rather, they are used instead of $o_i$ and $m_i$ for the sake of simplicity. We will now work on the complete-data log-likelihood, where the complete data includes observed and missing values. Equation ~\eqref{eq:lc_regr} becomes:
\begin{align}
      l_c (\bvartheta) = l_{c1} (\bvartheta_\bX) + l_{c2}(\bvartheta_\bY), \label{eq:lc_regr_missing}
\end{align}
where
\begin{align*}
    l_{c1} (\bvartheta_\bX) &= -\dfrac{1}{2} \sum_{i = 1}^n  \left[ d_\bX \log (2 \pi) + \log \lvert \bSigma_{\bX} \rvert + \left( \begin{bmatrix} \bx^o_i \\[1ex] \bx^m_i \end{bmatrix} - \begin{bmatrix} {\bmu}^o_{\bX} \\[1ex] \bmu^m_{\bX} \end{bmatrix} \right)^\top \bSigma^{-1}_{\bX} \left( \begin{bmatrix} \bx^o_i \\[1ex] \bx^m_i \end{bmatrix} - \begin{bmatrix} {\bmu}^o_{\bX} \\[1ex] {\bmu}^m_{\bX} \end{bmatrix} \right) \right] \\[1ex]
    l_{c2} (\bvartheta_\bY)&= -\dfrac{1}{2} \sum_{i = 1}^n \left[ d_\bY \log (2 \pi) + \log \lvert \bSigma_{\bY} \rvert + \left( \begin{bmatrix} \by^o_i \\[1ex] \by^m_i \end{bmatrix} - \bbeta^\top \begin{bmatrix} 1 \\ \bx^o_i \\[1ex] \bx^m_i \end{bmatrix} \right)^\top \bSigma^{-1}_{\bY} \left( \begin{bmatrix} \by^o_i \\[1ex] \by^m_i \end{bmatrix} - \bbeta^\top \begin{bmatrix} 1 \\ \bx^o_i \\[1ex] \bx^m_i \end{bmatrix} \right) \right].
\end{align*}

ML estimates of $\bvartheta$ in the presence of MAR values can be obtained using the expectation-maximization (EM) algorithm introduced by \citet{DemLai77}. 
This algorithm iteratively alternates between an expectation (E) step and a maximization (M) step until convergence. 
Details on the two steps are as follows.
\begin{description}
    \item[E-Step:] Compute the expectation $Q \left( \bvartheta ; \dot{\bvartheta} \right) = E \left[ \, l_c \left(\bvartheta \right) \mid \bx^o_1, \ldots, \bx^o_n, \by^o_1, \ldots, \by^o_n , \dot{\bvartheta} \right]$, where $l_c \left( \bvartheta\right)$ is the complete-data log-likelihood function based on $\{ \bX^o_i, \bX^m_i, \bY^o_i, \bY^m_i \}_{i = 1}^n$ and $\dot{\bvartheta}$ is the update of $\bvartheta$ from the previous M-step.
    \item[M-Step:] Update the current parameters $\dot{\bvartheta}$ with the new parameters $\ddot{\bvartheta}$ that maximize $Q \left( \bvartheta ; \dot{\bvartheta} \right)$.
\end{description}
Herein, the superscripts single dots and double dots on top of the parameters stand for estimates at the previous and current iterations,
respectively. From \eqref{eq:lc_regr_missing}, we can see that in the E-step, we need to compute the conditional expected value of the following terms given the observed data $\bx^o_i$ and $\by^o_i$: 
  \begin{align}\label{eq:emalx} \left( \begin{bmatrix} \bx^o_i \\[1ex] \bX^m_i \end{bmatrix} - \begin{bmatrix} \dot{\bmu}^o_{\bX} \\[1ex] \dot{\bmu}^m_{\bX} \end{bmatrix} \right)^\top \dot{\bSigma}^{-1}_{\bX} \left( \begin{bmatrix} \bx^o_i \\[1ex] \bX^m_i \end{bmatrix} - \begin{bmatrix} \dot{\bmu}^o_{\bX} \\[1ex] \dot{\bmu}^m_{\bX} \end{bmatrix} \right), \\[2ex] \label{eq:emaly}
  \left( \begin{bmatrix} \by^o_i \\[1ex] \bY^m_i \end{bmatrix} - \dot{\bbeta}^\top \begin{bmatrix} 1 \\ \bx^o_i \\[1ex] \bX^m_i \end{bmatrix} \right)^\top \dot{\bSigma}^{-1}_{\bY} \left( \begin{bmatrix} \by^o_i \\[1ex] \bY^m_i \end{bmatrix} - \dot{\bbeta}^\top \begin{bmatrix} 1 \\ \bx^o_i \\[1ex] \bX^m_i \end{bmatrix} \right).
    \end{align}
The previously defined propositions are needed for this computation, specifically, using Proposition~\ref{prop:cond_X2_x1_y1}, 
setting $\bX_2=\bX^m_i$, $\bx_1=\bx^o_i$, and $\by_1=\by^o_i$, the first expectation involving \eqref{eq:emalx} can be simplified as follows
\begin{align*}
    & E \left( \left( \begin{bmatrix} \bx^o_i \\[1ex] \bX^m_i \end{bmatrix} - \begin{bmatrix} \dot{\bmu}^o_{\bX} \\[1ex] \dot{\bmu}^m_{\bX} \end{bmatrix} \right)^\top \dot{\bSigma}^{-1}_{\bX} \left( \begin{bmatrix} \bx^o_i \\[1ex] \bX^m_i \end{bmatrix} - \begin{bmatrix} \dot{\bmu}^o_{\bX} \\[1ex] \dot{\bmu}^m_{\bX} \end{bmatrix} \right) \, \middle| \, \bx^o_i, \by^o_i, \dot{\bvartheta} \right) \\[1ex]
    =& \, \text{Trace} \left( \dot{\bSigma}^{-1}_{\bX} \left[ \ddot{\bOmega}_{\bX \mid i} + \left( \begin{bmatrix} \bx^o_i \\[1ex] \ddot{\tilde{\bx}}_{i} \end{bmatrix} - \begin{bmatrix} \dot{\bmu}^o_{\bX} \\[1ex] \dot{\bmu}^m_{\bX} \end{bmatrix} \right) \left( \begin{bmatrix} \bx^o_i \\[1ex]\ddot{ \tilde{\bx}}_{i} \end{bmatrix} - \begin{bmatrix} \dot{\bmu}^o_{\bX} \\[1ex] \dot{\bmu}^m_{\bX} \end{bmatrix} \right)^\top \right] \right),
\end{align*}
where
\begin{align*}
    \ddot{\tilde{\bx}}_{i} &= E \left( \bX^m_i \mid \bx^o_i, \by^o_i, \dot{\bvartheta} \right)\\
    \ddot{\bOmega}_{\bX \mid i} &= \text{Cov} \left( \begin{bmatrix} \bx^o_i \\[1ex] \bX^m_i \end{bmatrix} \, \middle| \, \bx^o_i, \by^o_i, \dot{\bvartheta} \right) = \begin{bmatrix} \boldsymbol{0}_{d^o_\bX \times d^o_\bX} & \boldsymbol{0}_{d^o_\bX \times d^m_\bX} \\[2ex]  \boldsymbol{0}_{d^m_\bX \times d^o_\bX} & \text{Cov} \left( \bX^m_i \mid \bx^o_i, \by^o_i, \dot{\bvartheta} \right) \end{bmatrix}, \\
    \text{and} \quad \text{Cov} \left( \bX^m_i \mid \bx^o_i, \by^o_i, \dot{\bvartheta} \right) &= \, E \left( \bX^m_i \bX^{m \top}_i \mid \bx^o_i, \by^o_i, \dot{\bvartheta} \right) - E \left( \bX^m_i \mid \bx^o_i, \by^o_i, \dot{\bvartheta} \right) E \left( \bX^m_i \mid \bx^o_i, \by^o_i, \dot{\bvartheta} \right)^\top.
\end{align*}

\begin{proposition} \label{prop:xtildeR}For missing covariates, the conditional expectations of $\bX^m_{i}$ and $\bX^m_i \bX^{m\top}_i$ given observed covariates $\bx^o_i$, observed responses $\by^o_i$, and current parameter estimates $\dot{\bvartheta}$ are obtained by
\begin{align*}
    \ddot{\tilde{\bx}}_i &= E \left( \bX^m_i \mid \bx^o_i, \by^o_i, \dot{\bvartheta} \right) = \dot{\bmu}^m_{\bX  } + \begin{bmatrix} \dot{\bSigma}^{mo}_{\bX \bX  } & \dot{\bSigma}^{mo}_{\bX \bY  } \end{bmatrix} \begin{bmatrix} \dot{\bSigma}^{oo}_{\bX \bX  } & \dot{\bSigma}^{oo}_{\bX \bY  } \\[1ex] \dot{\bSigma}^{oo}_{\bY \bX  } & \dot{\bSigma}^{oo}_{\bY \bY  } \end{bmatrix}^{-1} \left( \begin{bmatrix} \bx^o_i \\[1ex] \by^o_i \end{bmatrix} - \begin{bmatrix} \dot{\bmu}^o_{\bX  } \\[1ex] \dot{\tilde{\bmu}}^o_{\bY  } \end{bmatrix} \right)
\end{align*}
and
\begin{align*}
    \ddot{\tilde{\tilde{\bx}}}_i &= E \left( \bX^m_i \bX^{m \top}_i \mid \bx^o_i, \by^o_i , \dot{\bvartheta} \right) = \dot{\bSigma}^{mm}_{\bX \bX  } - \begin{bmatrix} \dot{\bSigma}^{mo}_{\bX \bX  } & \dot{\bSigma}^{mo}_{\bX \bY  } \end{bmatrix} \begin{bmatrix} \dot{\bSigma}^{oo}_{\bX \bX  } & \dot{\bSigma}^{oo}_{\bX \bY  } \\[1ex] \dot{\bSigma}^{oo}_{\bY \bX  } & \dot{\bSigma}^{oo}_{\bY \bY  } \end{bmatrix}^{-1} \begin{bmatrix} \dot{\bSigma}^{om}_{\bX \bX  } \\[1ex] \dot{\bSigma}^{om}_{\bY \bX  } \end{bmatrix} + \ddot{\tilde{\bx}}_i \ddot{\tilde{\bx}}^\top_i
\end{align*}
respectively.
\end{proposition}
\begin{proof}
Using Proposition~\ref{prop:cond_X2_x1_y1} 
with $\bX_1=\bX^o_i$, and $\bY_1=\bY^o_i$, we have the conditional distribution 
\begin{align*}
    \bX^m_i \mid \bx^o_i, \by^o_i , \dot{\bvartheta} \sim  \mathcal{G} \left( \bmu^\ast_i, \bSigma^\ast_i \right),
\end{align*}
where
\begin{align*}
    \bmu^\ast_i &= E \left( \bX^m_i \mid \bx^o_i, \by^o_i , \dot{\bvartheta} \right) = \dot{\bmu}^m_{\bX  } + \begin{bmatrix} \dot{\bSigma}^{mo}_{\bX \bX  } & \dot{\bSigma}^{mo}_{\bX \bY  } \end{bmatrix} \begin{bmatrix} \dot{\bSigma}^{oo}_{\bX \bX  } & \dot{\bSigma}^{oo}_{\bX \bY  } \\[1ex] \dot{\bSigma}^{oo}_{\bY \bX  } & \dot{\bSigma}^{oo}_{\bY \bY  } \end{bmatrix}^{-1} \left( \begin{bmatrix} \bx^o_i \\[1ex] \by^o_i \end{bmatrix} - \begin{bmatrix} \dot{\bmu}^o_{\bX  } \\[1ex] \dot{\tilde{\bmu}}^o_{\bY  } \end{bmatrix} \right)
\end{align*}
and
\begin{align*}
    \bSigma^\ast_i &= \text{Cov} \left( \bX^m_i \mid \bx^o_i, \by^o_i , \dot{\bvartheta} \right) = \dot{\bSigma}^{mm}_{\bX \bX  } - \begin{bmatrix} \dot{\bSigma}^{mo}_{\bX \bX  } & \dot{\bSigma}^{mo}_{\bX \bY  } \end{bmatrix} \begin{bmatrix} \dot{\bSigma}^{oo}_{\bX \bX  } & \dot{\bSigma}^{oo}_{\bX \bY  } \\[1ex] \dot{\bSigma}^{oo}_{\bY \bX  } & \dot{\bSigma}^{oo}_{\bY \bY  } \end{bmatrix}^{-1} \begin{bmatrix} \dot{\bSigma}^{om}_{\bX \bX  } \\[1ex] \dot{\bSigma}^{om}_{\bY \bX  } \end{bmatrix}.
\end{align*}
Note that
\begin{align*}
    \ddot{\tilde{\bx}}_i = E \left( \bX^m_i \mid \bx^o_i, \by^o_i , \dot{\bvartheta} \right) = \bmu^\ast_i \quad \text{and} \quad \ddot{\tilde{\tilde{\bx}}}_i = E \left( \bX^m_i \bX^{m \top}_i \mid \bx^o_i, \by^o_i , \dot{\bvartheta} \right) = \bSigma^\ast_{ij} + \ddot{\tilde{\bx}}_i \ddot{\tilde{\bx}}^\top_i.
\end{align*}
Accordingly, plugging the values of $\bmu^\ast_i$ and $\bSigma^\ast_i$ into $\ddot{\tilde{\bx}}_i$ and $\ddot{\tilde{\tilde{\bx}}}_i$ completes the proof.
\end{proof}

It is interesting to see that in the computation of both $\ddot{\tilde{\bx}}_i$ and $\ddot{\tilde{\tilde{\bx}}}_i$, the terms added to $\dot{\bmu}^{m}_{\bX}$ and $\dot{\bSigma}^{mm}_{\bX \bX}$, respectively, represent an adjustment for imputing the conditions in the expectation
computation.\\
Similarly, from Propositions \ref{prop:cond_Y2_X1_Y1}  
and \ref{prop:cond_Y2_X2_X1_Y1}, setting $\bX_2=\bX^m_i$, $\bY_2=\bY^m_i$, $\bx_1=\bx^o_i$, and $\by_1=\by^o_i$, the second expectation involving \eqref{eq:emaly} can be simplified as follows
\begin{align*}
    & E \left( \left( \begin{bmatrix} \by^o_i \\[1ex] \bY^m_i \end{bmatrix} - \dot{\bbeta}^\top \begin{bmatrix} 1 \\ \bx^o_i \\[1ex] \bX^m_i \end{bmatrix} \right)^\top \dot{\bSigma}^{-1}_{\bY} \left( \begin{bmatrix} \by^o_i \\[1ex] \bY^m_i \end{bmatrix} - \dot{\bbeta}^\top \begin{bmatrix} 1 \\ \bx^o_i \\[1ex] \bX^m_i \end{bmatrix} \right) \, \middle| \, \bx^o_i, \by^o_i, \dot{\bvartheta} \right) \\[1ex]
    =& \,  \text{Trace} \left[ \dot{\bSigma}^{-1}_{\bY} \left( \ddot{\bOmega}_{\bY \mid i} + \dot{\bDelta}_{i} \dot{\bbeta} + \dot{\bbeta}^\top \dot{\bDelta}^\top_{i} + \dot{\bbeta}^\top \ddot{\bOmega}^\ast_{\bX \mid i} \right) \right] + \left( \begin{bmatrix} \by^o_i \\[1ex] \ddot{\tilde{\by}}_{i} \end{bmatrix} - \dot{\bbeta}^\top \begin{bmatrix} 1 \\ \bx^o_i \\[1ex] \ddot{\tilde{\bx}}_{i} \end{bmatrix} \right)^\top \dot{\bSigma}^{-1}_{\bY} \left( \begin{bmatrix} \by^o_i \\[1ex] \ddot{\tilde{\by}}_{i} \end{bmatrix} - \dot{\bbeta}^\top \begin{bmatrix} 1 \\ \bx^o_i \\[1ex] \ddot{\tilde{\bx}}_{i} \end{bmatrix} \right).
\end{align*}
Note that this expectation involves many quantities, which are defined as follows
\begin{align*}
    \ddot{\bOmega}_{\bY \mid i} &= \text{Cov} \left( \begin{bmatrix} \by^o_i \\[1ex] \bY^m_i \end{bmatrix} \, \middle| \, \bx^o_i, \by^o_i, \dot{\bvartheta} \right) = \begin{bmatrix} \boldsymbol{0}_{d^o_\bX \times d^o_\bX} & \boldsymbol{0}_{d^o_\bY \times d^m_\bY} \\[2ex]  \boldsymbol{0}_{d^m_\bY \times d^o_\bY} & \ddot{\tilde{\tilde{\by}}}_{i} - \ddot{\tilde{\by}}_{i} \ddot{\tilde{\by}}^\top_{i} \end{bmatrix}, \\[2ex]
    \ddot{\tilde{\by}}_i &= E \left( \bY^m_i \mid \bx^o_i, \by^o_i, \dot{\bvartheta} \right), \quad  \ddot{ \tilde{\tilde{\by}}}_i = E \left( \bY^m_i \bY^{m \top}_i \mid \bx^o_i, \by^o_i, \dot{\bvartheta} \right),  \\[2ex]
    \ddot{\bOmega}^\ast_{\bX \mid i} &= \text{Cov} \left( \begin{bmatrix} 1 \\ \bx^o_i \\[1ex] \bX^m_i \end{bmatrix} \, \middle| \, \bx^o_i, \by^o_i, \dot{\bvartheta} \right) = \begin{bmatrix} \boldsymbol{0}_{1 \times 1} & \boldsymbol{0}_{1 \times d^o_\bX} & \boldsymbol{0}_{1 \times d^m_\bX} \\[2ex] \boldsymbol{0}_{d^o_\bX \times 1} &  \boldsymbol{0}_{d^o_\bX \times d^o_\bX} & \boldsymbol{0}_{d^o_\bX \times d^m_\bX} \\[2ex]  \boldsymbol{0}_{d^m_\bX \times 1} & \boldsymbol{0}_{d^m_\bX \times d^o_\bX} & \ddot{\tilde{\tilde{\bx}}}_{i} - \ddot{\tilde{\bx}}_{i} \ddot{\tilde{\bx}}^\top_{i} \end{bmatrix} = \begin{bmatrix} \boldsymbol{0}_{1 \times 1} & \boldsymbol{0}_{1 \times d_\bX} \\[2ex] \boldsymbol{0}_{d_\bX \times 1} &  \ddot{\bOmega}_{\bX \mid i} \end{bmatrix}, \\[2ex]
    \dot{\bDelta}_{i} &= \text{Cov} \left( \begin{bmatrix} \by^o_i \\[1ex] \bY^m_i \end{bmatrix}, \begin{bmatrix} 1 \\ \bx^o_i \\[1ex] \bX^m_i \end{bmatrix} \, \middle| \, \bx^o_i, \by^o_i, \dot{\bvartheta} \right) = \begin{bmatrix}
\boldsymbol{0}_{d^o_\bY \times 1} & \boldsymbol{0}_{d^o_\bY \times d^o_\bX} & \boldsymbol{0}_{d^o_\bY \times d^m_\bX}  \\[2ex]
\boldsymbol{0}_{d^m_\bY \times 1} & \boldsymbol{0}_{d^m_\bY \times d^o_\bX} & \dot{\tilde{\bSigma}}^{mm}_{\bY \bX \mid i}
\end{bmatrix}, \\[2ex]
    \text{and} \quad \dot{\tilde{\bSigma}}^{mm}_{\bY \bX \mid i} &= \text{Cov} \left( \bY^m_i, \bX^m_i \mid \bx^o_i, \by^o_i, \dot{\bvartheta} \right) \\
    &= \begin{bmatrix} \dot{\bSigma}^{mm}_{\bY \bY} & \dot{\bSigma}^{mm}_{\bY \bX} \\[1ex] \dot{\bSigma}^{mm}_{\bX \bY} & \dot{\bSigma}^{mm}_{\bX \bX} \end{bmatrix} - \begin{bmatrix} \dot{\bSigma}^{mo}_{\bY \bY} & \dot{\bSigma}^{mo}_{\bY \bX} \\[1ex]
\dot{\bSigma}^{mo}_{\bX \bY} & \dot{\bSigma}^{mo}_{\bX \bX} \end{bmatrix} \begin{bmatrix} \dot{\bSigma}^{oo}_{\bY \bY} & \dot{\bSigma}^{oo}_{\bY \bX} \\[1ex] \dot{\bSigma}^{oo}_{\bX \bY} & \dot{\bSigma}^{oo}_{\bX \bX} \end{bmatrix}^{-1} \begin{bmatrix} \dot{\bSigma}^{om}_{\bY \bY} & \dot{\bSigma}^{om}_{\bY \bX} \\[1ex] \dot{\bSigma}^{om}_{\bX \bY} & \dot{\bSigma}^{om}_{\bX \bX} \end{bmatrix}.
\end{align*}

These expectations can be computed using the following propositions.
\begin{proposition}  
\label{prop:ytildeR} 
For missing responses, the conditional expectations of $\bY^m_{i}$ and $\bY^m_i \bY^{m\top}_i$ given observed covariates $\bx^o_i$, observed responses $\by^o_i$, and current parameter estimates $\dot{\bvartheta}$ are obtained by
\begin{align*}
    \ddot{\tilde{\by}}_{i} &= E \left( \bY^m_i \mid \bx^o_i, \by^o_i , \dot{\bvartheta} \right) = \dot{\tilde{\bmu}}^m_{\bY  } + \begin{bmatrix} \dot{\bSigma}^{mo}_{\bY \bX  } & \dot{\bSigma}^{mo}_{\bY \bY  } \end{bmatrix} \begin{bmatrix} \dot{\bSigma}^{oo}_{\bX \bX  } & \dot{\bSigma}^{oo}_{\bX \bY  } \\[1ex] \dot{\bSigma}^{oo}_{\bY \bX  } & \dot{\bSigma}^{oo}_{\bY \bY  } \end{bmatrix}^{-1} \left( \begin{bmatrix} \bx^o_i \\[1ex] \by^o_i \end{bmatrix} - \begin{bmatrix} \dot{\bmu}^o_{\bX  } \\[1ex] \dot{\tilde{\bmu}}^o_{\bY  } \end{bmatrix} \right)
\end{align*}
and
\begin{align*}
    \ddot{ \tilde{\tilde{\by}}}_{i} &= E \left( \bY^m_i \bY^{m \top}_i \mid \bx^o_i, \by^o_i , \dot{\bvartheta} \right) = \dot{\bSigma}^{mm}_{\bY \bY  } - \begin{bmatrix} \dot{\bSigma}^{mo}_{\bY \bX  } & \dot{\bSigma}^{mo}_{\bY \bY  } \end{bmatrix} \begin{bmatrix} \dot{\bSigma}^{oo}_{\bX \bX  } & \dot{\bSigma}^{oo}_{\bX \bY  } \\[1ex] \dot{\bSigma}^{oo}_{\bY \bX  } & \dot{\bSigma}^{oo}_{\bY \bY  } \end{bmatrix}^{-1} \begin{bmatrix} \dot{\bSigma}^{om}_{\bX \bY  } \\[1ex] \dot{\bSigma}^{om}_{\bY \bY  } \end{bmatrix} + \ddot{\tilde{\by}}_{i} \ddot{\tilde{\by}}^\top_{i}
\end{align*}
respectively.
\end{proposition}
\begin{proof}
Using proposition~\ref{prop:cond_Y2_X1_Y1} 
with $\bX_1=\bX^o_i$, and $\bY_1=\bY^o_i$, we have the conditional distribution 
\begin{align*}
    \bY^m_i \mid \bx^o_i, \by^o_i , \dot{\bvartheta} \sim  \mathcal{G} \left( \bmu^{\ast \ast}_{i}, \bSigma^{\ast \ast}_{i} \right),
\end{align*}
where
\begin{align*}
    \bmu^{\ast \ast}_{i} &= E \left( \bY^m_i \mid \bx^o_i, \by^o_i , \dot{\bvartheta} \right) = \dot{\tilde{\bmu}}^m_{\bY  } + \begin{bmatrix} \dot{\bSigma}^{mo}_{\bY \bX  } & \dot{\bSigma}^{mo}_{\bY \bY  } \end{bmatrix} \begin{bmatrix} \dot{\bSigma}^{oo}_{\bX \bX  } & \dot{\bSigma}^{oo}_{\bX \bY  } \\[1ex] \dot{\bSigma}^{oo}_{\bY \bX  } & \dot{\bSigma}^{oo}_{\bY \bY  } \end{bmatrix}^{-1} \left( \begin{bmatrix} \bx^o_i \\[1ex] \by^o_i \end{bmatrix} - \begin{bmatrix} \dot{\bmu}^o_{\bX  } \\[1ex] \dot{\tilde{\bmu}}^o_{\bY  } \end{bmatrix} \right) 
\end{align*}
and
\begin{align*}
    \bSigma^{\ast \ast}_{i} &= \text{Var} \left( \bY^m_i \mid \bx^o_i, \by^o_i , \dot{\bvartheta} \right) = \dot{\bSigma}^{mm}_{\bY \bY  } - \begin{bmatrix} \dot{\bSigma}^{mo}_{\bY \bX  } & \dot{\bSigma}^{mo}_{\bY \bY  } \end{bmatrix} \begin{bmatrix} \dot{\bSigma}^{oo}_{\bX \bX  } & \dot{\bSigma}^{oo}_{\bX \bY  } \\[1ex] \dot{\bSigma}^{oo}_{\bY \bX  } & \dot{\bSigma}^{oo}_{\bY \bY  } \end{bmatrix}^{-1} \begin{bmatrix} \dot{\bSigma}^{om}_{\bX \bY  } \\[1ex] \dot{\bSigma}^{om}_{\bY \bY  } \end{bmatrix} + \ddot{\tilde{\by}}_{i} \ddot{\tilde{\by}}^\top_{i}.
\end{align*}
Note that
\begin{align*}
    \ddot{\tilde{\by}}_{i} = E \left( \bY^m_i \mid \bx^o_i, \by^o_i , \dot{\bvartheta} \right) = \bmu^\ast_{i} \quad \text{and} \quad \ddot{\tilde{\tilde{\by}}}_{i} = E \left( \bY^m_i \bY^{m \top}_i \mid \bx^o_i, \by^o_i , \dot{\bvartheta} \right) = \bSigma^{\ast \ast}_{i} + \ddot{\tilde{\by}}_{i} \ddot{\tilde{\by}}^\top_{i}.
\end{align*}
Accordingly, plugging the values of $\bmu^{\ast \ast}_{i}$ and $\bSigma^{\ast \ast}_{i}$ into $\ddot{\tilde{\by}}_{i}$ and $\ddot{\tilde{\tilde{\by}}}_{i}$ completes the proof.
\end{proof}

\begin{proposition}  \label{prop:cov}
Given observed covariates $\bx^o_i$ and observed responses $\by^o_i$, the covariance between missing responses $\bY^m_i$ and missing covariates $\bX^m_i$ is given by
\begin{align*}
    \ddot{\tilde{\bSigma}}^{mm}_{\bY \bX \mid i } &= \text{Cov} \left( \bY^m_i, \bX^m_i \, \middle| \, \bx^o_i, \by^o_i , \dot{\bvartheta} \right) \\
    &= \begin{bmatrix} \dot{\bSigma}^{mm}_{\bY \bY  } & \dot{\bSigma}^{mm}_{\bY \bX  } \\[1ex] \dot{\bSigma}^{mm}_{\bX \bY  } & \dot{\bSigma}^{mm}_{\bX \bX  } \end{bmatrix} - \begin{bmatrix} \dot{\bSigma}^{mo}_{\bY \bY  } & \dot{\bSigma}^{mo}_{\bY \bX  } \\[1ex]
\dot{\bSigma}^{mo}_{\bX \bY  } & \dot{\bSigma}^{mo}_{\bX \bX  } \end{bmatrix} \begin{bmatrix} \dot{\bSigma}^{oo}_{\bY \bY  } & \dot{\bSigma}^{oo}_{\bY \bX  } \\[1ex] \dot{\bSigma}^{oo}_{\bX \bY  } & \dot{\bSigma}^{oo}_{\bX \bX  } \end{bmatrix}^{-1} \begin{bmatrix} \dot{\bSigma}^{om}_{\bY \bY  } & \dot{\bSigma}^{om}_{\bY \bX  } \\[1ex] \dot{\bSigma}^{om}_{\bX \bY  } & \dot{\bSigma}^{om}_{\bX \bX  } \end{bmatrix}.
\end{align*}
\end{proposition}
\begin{proof}
Using proposition~\ref{prop:cond_Y2_X2_X1_Y1}
with $\bX_1=\bX^o_i$, and $\bY_1=\bY^o_i$, we have the conditional distribution 
\begin{align*}
    \bY^m_i, \bX^m_i \mid \bx^o_i, \by^o_i , \dot{\bvartheta} \sim  \mathcal{G} \left( \bmu^{\ast \ast \ast}_{i }, \bSigma^{\ast \ast \ast}_{i } \right),
\end{align*}
where
\begin{align*}
    \bmu^{\ast \ast \ast}_{i } &= \begin{bmatrix} \tilde{\dot{\bmu}}^m_{\bY  } \\ \dot{\bmu}^m_{\bX  } \end{bmatrix} + \begin{bmatrix} \dot{\bSigma}^{mo}_{\bY \bY  } & \dot{\bSigma}^{mo}_{\bY \bX  } \\
\dot{\bSigma}^{mo}_{\bX \bY  } & \dot{\bSigma}^{mo}_{\bX \bX  } \end{bmatrix} \begin{bmatrix} \dot{\bSigma}^{oo}_{\bY \bY  } & \dot{\bSigma}^{oo}_{\bY \bX  } \\ \dot{\bSigma}^{oo}_{\bX \bY  } & \dot{\bSigma}^{oo}_{\bX \bX  } \end{bmatrix}^{-1} \left( \begin{bmatrix} \bx^o_i \\ \by^o_i \end{bmatrix} - \begin{bmatrix} \dot{\bmu}^o_{\bX  } \\ \dot{\tilde{\bmu}}^o_{\bY  } \end{bmatrix} \right) \\[1ex]
    \bSigma^{\ast \ast \ast}_{i } &= \begin{bmatrix} \dot{\bSigma}^{mm}_{\bY \bY  } & \dot{\bSigma}^{mm}_{\bY \bX  } \\[1ex] \dot{\bSigma}^{mm}_{\bX \bY  } & \dot{\bSigma}^{mm}_{\bX \bX  } \end{bmatrix} - \begin{bmatrix} \dot{\bSigma}^{mo}_{\bY \bY  } & \dot{\bSigma}^{mo}_{\bY \bX  } \\[1ex]
\dot{\bSigma}^{mo}_{\bX \bY  } & \dot{\bSigma}^{mo}_{\bX \bX  } \end{bmatrix} \begin{bmatrix} \dot{\bSigma}^{oo}_{\bY \bY  } & \dot{\bSigma}^{oo}_{\bY \bX  } \\[1ex] \dot{\bSigma}^{oo}_{\bX \bY  } & \dot{\bSigma}^{oo}_{\bX \bX  } \end{bmatrix}^{-1} \begin{bmatrix} \dot{\bSigma}^{om}_{\bY \bY  } & \dot{\bSigma}^{om}_{\bY \bX  } \\[1ex] \dot{\bSigma}^{om}_{\bX \bY  } & \dot{\bSigma}^{om}_{\bX \bX  } \end{bmatrix}
\end{align*}
Therefore,
\begin{align*}
    \ddot{\tilde{\bSigma}}^{mm}_{\bY \bX \mid i } = \text{Cov} \left( \bY^m_i, \bX^m_i \mid \bx^o_i, \by^o_i , \dot{\bvartheta} \right) = \bSigma^{\ast \ast \ast}_{i }.
\end{align*}
\end{proof}

Similarly to what we observed for the predictors, in the computation of both $\ddot{\tilde{\by}}_i$ and $\ddot{\tilde{\tilde{\by}}}_i$, the terms added to $\dot{\bmu}^{m}_{\bY}$ and $\dot{\bSigma}^{mm}_{\bY \bY}$, respectively, represent an adjustment for imputing the conditions in the expectation computation. While $\bY_i^m$ does not appear in the first expectation, both missing responses and covariates have a crucial role in the second expectation. This is evident from the fact that the covariance matrix $\text{Cov} (\bY_i^m, \bX_i^m | \bx_i^o, \by_i^o, \dot{\bvartheta})$ is needed for this computation. In other words, while the computation of the expectation of $l_{c1}(\bvartheta_{\bX})$ involves only $\ddot{\bOmega}_{\bX \mid i}$, the expectation of $l_{c2}(\bvartheta_{\bY})$ depends on $\ddot{\bOmega}^\ast_{\bX \mid i}$,  $\ddot{\bOmega}_{\bY\mid i}$, and the covariance $\text{Cov} (\bY_i^m, \bX_i^m | \bx_i^o, \by_i^o, \dot{\bvartheta})$, emphasizing the different role played by responses and covariates.

In the M-step, the parameters are updated as
\begin{align*}
    \ddot{\bmu}_{\bX } &= \dfrac{1}{n} \sum_{i = 1}^n \begin{bmatrix} \bx^o_i \\[1ex] \ddot{\tilde{\bx}}_i \end{bmatrix}, \\[2ex]
    \ddot{\bSigma}_{\bX } &= \dfrac{1}{n} \sum_{i = 1}^n \tilde{z}_i \left[ \ddot{\bOmega}_{\bX \mid i}  + \left( \begin{bmatrix} \bx^o_i \\[1ex] \ddot{\tilde{\bx}}_i \end{bmatrix} - \ddot{\bmu}_{\bX } \right) \left( \begin{bmatrix} \bx^o_i \\[1ex] \ddot{\tilde{\bx}}_i \end{bmatrix} - \ddot{\bmu}_{\bX } \right)^\top \right], \\[2ex]
    \ddot{\bbeta }&= \left( \sum_{i = 1}^n \begin{bmatrix} 1 \\ \bx^o_i \\[1ex] \ddot{\tilde{\bx}}_i \end{bmatrix} \begin{bmatrix} 1 & \bx^o_i & \ddot{\tilde{\bx}}_i \end{bmatrix} + \sum_{i = 1}^n \ddot{\bOmega}^\ast_{\bX \mid i} \right)^{-1} \left( \sum_{i = 1}^n \begin{bmatrix} 1 \\ \bx^o_i \\[1ex] \ddot{\tilde{\bx}}_i \end{bmatrix} \begin{bmatrix} \by^o_i & \ddot{\tilde{\by}}_i \end{bmatrix} + \sum_{i = 1}^n \dot{\bDelta}^\top_i \right), \\[2ex]
   \ddot{ \bSigma}_{\bY } &= \dfrac{1}{n} \sum_{i = 1}^n \left[ \left( \ddot{\bOmega}_{\bY \mid i} + \dot{\bDelta}_i \ddot{\bbeta} + \ddot{\bbeta}^\top \dot{\bDelta}^\top_i + \ddot{\bbeta}^\top \ddot{\bOmega}^\ast_{\bX \mid i} \ddot{\bbeta} \right)^\top + \left( \begin{bmatrix} \by^o_i \\[1ex] \ddot{\tilde{\by}}_i \end{bmatrix} - \ddot{\bbeta}^\top \begin{bmatrix} 1 \\ \bx^o_i \\[1ex] \ddot{\tilde{\bx}}_i \end{bmatrix} \right) \left( \begin{bmatrix} \by^o_i \\[1ex] \ddot{\tilde{\by}}_i \end{bmatrix} - \ddot{\bbeta}^\top \begin{bmatrix} 1 \\ \bx^o_i \\[1ex] \ddot{\tilde{\bx}}_i \end{bmatrix} \right)^\top \right].
\end{align*}
As in the E-step, we again observe that covariates and responses play different roles in the M-step. Specifically, the computation of $\ddot{\bmu}_{\bX }$ and $\ddot{\bSigma}_{\bX }$ depends on the covariates, whereas the computation of $\ddot{\bbeta}$ and $\ddot{\bSigma}_{\bY }$ involves both covariates and responses.

\section{Cluster-wise multivariate linear regression with multiple random covariates and missing values}\label{sec:clusterwise_missing}


Model-based clustering is a statistical methodology that assumes observations arise from a finite mixture of probability distributions, with each component typically corresponding to a distinct group or cluster in the data. 
A natural extension of model~\eqref{eq:jointreg} within this framework is the mixture of regressions with random covariates (MRRC), which models the joint distribution of $\left(\bX^\top, \bY^\top\right)^\top$ as follows:
\begin{equation}
    p_\text{MRRC} (\bx, \by; \bvartheta) = \sum_{j = 1}^k \pi_j \, \phi\left( \by; \bmu_{\bY \mid j} (\bx; \bbeta_j), \bSigma_{\bY \mid j} \right) \, \phi\left( \bx; \bmu_{\bX \mid j}, \bSigma_{\bX \mid j} \right),
    \label{eq:G-MRRC}
\end{equation}
where $\pi_j > 0$ and $\sum_{j = 1}^k \pi_j = 1$ are the mixing proportions, with all the other parameters defined as in \eqref{eq:jointreg} but separately for each mixture component.
For a detailed treatment of its multivariate-multiple formulation, see \citet{dang2017multivariate}. 
This model is also referred to in the literature as the cluster-weighted model.

\subsection{Maximum likelihood estimates of MRRC}

To obtain the ML estimates of $\bvartheta$ for model~\eqref{eq:G-MRRC}, it is a common practice, in the mixture framework, to employ the expectation-maximization (EM) algorithm. 
Under this approach, the observed data $\left(\bx_1^\top,\by_1^\top\right)^\top,\ldots,\left(\bx_n^\top,\by_n^\top\right)^\top$ are treated as incomplete. 
This incompleteness stems from the unknown cluster memberships of the observations. Specifically, for the $i$th unit, this missing information is captured by the unobserved component membership vector $\bZ_i = \left( Z_{i1}, \ldots, Z_{ik} \right)^\top$, where $Z_{ij} = 1$ if $\left(\bx_i^\top,\by_i^\top\right)^\top$ originates from the $j$th mixture component, and $Z_{ij} = 0$ otherwise.
Therefore, the complete data become $\left(\bx_1^\top,\by_1^\top,\bz_1^\top\right)^\top,\ldots,\left(\bx_n^\top,\by_n^\top,\bz_n^\top\right)^\top$.
When MAR values are introduced, the complete-data log-likelihood function that the EM algorithm iterates on becomes
\begin{align*}
    l_c (\bvartheta) = l_{c1}(\bpi)  + l_{c2}(\bvartheta_\bX)
    + l_{c3}(\bvartheta_\bY),
\end{align*}
with $\bvartheta_\bX = ( \bvartheta_{\bX \mid 1}, \dots, \bvartheta_{\bX \mid k})$ and $\bvartheta_\bY = ( \bvartheta_{\bY \mid 1}, \dots, \bvartheta_{\bY \mid k})$  where 
\begin{align*}
    l_{c1}(\bpi)&= \sum_{i = 1}^n \sum_{j = 1}^k z_{ij} \log (\pi_j),\\[1ex]
    l_{c2}(\bvartheta_\bX)&= -\dfrac{1}{2} \sum_{i = 1}^n \sum_{j = 1}^k z_{ij} \left[ d_\bX \log (2 \pi) + \log \lvert \bSigma_{\bX \mid j} \rvert + \left( \begin{bmatrix} \bx^o_i \\[1ex] \bx^m_i \end{bmatrix} - \begin{bmatrix} {\bmu}^o_{\bX \mid j} \\ \bmu^m_{\bX \mid j} \end{bmatrix} \right)^\top \bSigma^{-1}_{\bX \mid j} \left( \begin{bmatrix} \bx^o_i \\[1ex] \bx^m_i \end{bmatrix} - \begin{bmatrix} {\bmu}^o_{\bX \mid j} \\ \bmu^m_{\bX \mid j} \end{bmatrix} \right) \right], \\[1ex]
    l_{c3}(\bvartheta_\bY)  &= -\dfrac{1}{2}  \sum_{i = 1}^n \sum_{j = 1}^k z_{ij} \left[ d_\bY \log (2 \pi) + \log \lvert \bSigma_{\bY \mid j} \rvert + \left( \begin{bmatrix} \by^o_i \\[1ex] \by^m_i \end{bmatrix} - \bbeta^\top_j \begin{bmatrix} 1 \\ \bx^o_i \\[1ex] \bx^m_i \end{bmatrix} \right)^\top \bSigma^{-1}_{\bY \mid j} \left( \begin{bmatrix} \by^o_i \\[1ex] \by^m_i \end{bmatrix} - \bbeta^\top_j \begin{bmatrix} 1 \\ \bx^o_i \\[1ex] \bx^m_i \end{bmatrix} \right) \right].\\
\end{align*}
Details on the steps of the EM algorithm are given below.
\begin{itemize}
\item \textbf{E-Step:} Compute the expectation $Q \left( \bvartheta ; \dot{\bvartheta} \right) = E \left[ \, l_c \left(\bvartheta \right) \mid \bx^o_1, \ldots, \bx^o_n, \by^o_1, \ldots, \by^o_n , \dot{\bvartheta} \right]$, where $l_c \left( \bvartheta\right)$ is the complete-data log-likelihood function which is now based on $\{ \bZ_i, \bX^o_i, \bX^m_i, \bY^o_i, \bY^m_i \}_{i = 1}^n$ and $\dot{\bvartheta} $ is the estimate of $\bvartheta$ at the previous iteration, being $\bvartheta$ the vector containing all the model parameters.    
\item \textbf{M-Step:} Update the current parameters $\dot{\bvartheta}$ with the new parameters $\ddot{\bvartheta}$ that maximize the expectation $Q \left( \bvartheta ; \dot{\bvartheta} \right)$.
\end{itemize}

The E-step involves the following expectations
\begin{align*}
       \ddot{\tilde{z}}_{ij} & = E \left( Z_{ij} \mid \bx^o_i, \by^o_i, \dot{\bvartheta} \right), \\
    \ddot{\tilde{\bx}}_{ij} &= E \left( \bX^m_i \mid \bx^o_i, \by^o_i, Z_{ij} = 1, \dot{\bvartheta} \right), \quad \ddot{\tilde{\tilde{\bx}}}_{ij} = E \left( \bX^m_i \bX^{m \top}_i \mid \bx^o_i, \by^o_i, Z_{ij} = 1, \dot{\bvartheta} \right), \\
    \ddot{\tilde{\by}}_{ij} &= E \left( \bY^m_i \mid \bx^o_i, \by^o_i, Z_{ij} = 1, \dot{\bvartheta} \right), \quad \ddot{ \tilde{\tilde{\by}}}_{ij} = E \left( \bY^m_i \bY^{m \top}_i \mid \bx^o_i, \by^o_i, Z_{ij} = 1, \dot{\bvartheta} \right) .
\end{align*}
In the next propositions, we derive closed-form solutions for the above expectations

\begin{proposition} The conditional expectation of $Z_{ij}$ given observed covariates $\bx^o_i$, observed responses $\by^o_i$, and current parameter estimates $\dot{\bvartheta}$ is obtained by
    \begin{align}
       \ddot{\tilde{z}}_{ij} & = E \left( Z_{ij} \mid \bx^o_i, \by^o_i, \dot{\bvartheta} \right) = \dfrac{\pi_j \phi \left( \begin{bmatrix} \bx^o_i \\[1ex] \by^o_i \end{bmatrix}; \begin{bmatrix} \dot{\bmu}^o_{\bX \mid j} \\[1ex] \dot{\tilde{\bmu}}^o_{\bY \mid j} \end{bmatrix}, \begin{bmatrix}
\dot{\bSigma}^{oo}_{\bX \bX \mid j} & \dot{\bSigma}^{oo}_{\bX \bY \mid j}  \\[1ex]
\dot{\bSigma}^{oo}_{\bY \bX \mid j} & \dot{\bSigma}^{oo}_{\bY \bY \mid j}
\end{bmatrix} \right)}{\displaystyle\sum_{j' = 1}^k \pi_{j'} \phi \left( \begin{bmatrix} \bx^o_i \\[1ex] \by^o_i \end{bmatrix}; \begin{bmatrix} \dot{\bmu}^o_{\bX \mid j'} \\[1ex] \dot{\tilde{\bmu}}^o_{\bY \mid j'} \end{bmatrix}, \begin{bmatrix}
\dot{\bSigma}^{oo}_{\bX \bX \mid j'} & \dot{\bSigma}^{oo}_{\bX \bY \mid j'}  \\[1ex]
\dot{\bSigma}^{oo}_{\bY \bX \mid j'} & \dot{\bSigma}^{oo}_{\bY \bY \mid j'}
\end{bmatrix} \right)}.     \label{eq:z_tilde}
     \end{align}
\end{proposition}
\begin{proof}
Using Proposition \ref{prop:joint_X1_Y1} 
with $\bX_1=\bX^o_i$, and $\bY_1=\bY^o_i$, we have the joint distribution
\begin{align*}
    \begin{bmatrix} \bX^o_i \\ \bY^o_i \end{bmatrix} \, \Bigg| \, Z_{ij} = 1 \sim  \mathcal{G} \left( \begin{bmatrix} \bmu^o_{\bX \mid j} \\ \tilde{\bmu}^o_{\bY \mid j} \end{bmatrix}, \begin{bmatrix}
\bSigma^{oo}_{\bX \bX \mid j} & \bSigma^{oo}_{\bX \bY \mid j}  \\
\bSigma^{oo}_{\bY \bX \mid j} & \bSigma^{oo}_{\bY \bY \mid j}
\end{bmatrix} \right).
\end{align*}
Since $E \left( Z_{ij} \mid \bx^o_i, \by^o_i, \dot{\bvartheta} \right) = P \left( Z_{ij} = 1 \mid \bx^o_i, \by^o_i, \dot{\bvartheta} \right)$, using Bayes' theorem, we obtain
\begin{align*}
   E \left( Z_{ij} \mid \bx^o_i, \by^o_i, \dot{\bvartheta} \right) &= \dfrac{P (Z_{ij} = 1; \dot{\bvartheta}) f \left( \bx^o_i, \by^o_i \mid Z_{ij} = 1, \dot{\bvartheta} \right)}{\displaystyle\sum_{j' = 1}^k P (Z_{ij'} = 1; \dot{\bvartheta}) f \left( \bx^o_i, \by^o_i \mid Z_{ij'} = 1, \dot{\bvartheta} \right)} \\
   &= \dfrac{\pi_j \phi \left( \begin{bmatrix} \bx^o_i \\ \by^o_i \end{bmatrix}; \begin{bmatrix} \bmu^o_{\bX \mid j} \\ \tilde{\bmu}^o_{\bY \mid j} \end{bmatrix}, \begin{bmatrix}
\bSigma^{oo}_{\bX \bX \mid j} & \bSigma^{oo}_{\bX \bY \mid j}  \\
\bSigma^{oo}_{\bY \bX \mid j} & \bSigma^{oo}_{\bY \bY \mid j}
\end{bmatrix} \right)}{\displaystyle\sum_{j' = 1}^k \pi_{j'} \phi \left( \begin{bmatrix} \bx^o_i \\ \by^o_i \end{bmatrix}; \begin{bmatrix} \bmu^o_{\bX \mid j'} \\ \tilde{\bmu}^o_{\bY \mid j'} \end{bmatrix}, \begin{bmatrix}
\bSigma^{oo}_{\bX \bX \mid j'} & \bSigma^{oo}_{\bX \bY \mid j'}  \\
\bSigma^{oo}_{\bY \bX \mid j'} & \bSigma^{oo}_{\bY \bY \mid j'}
\end{bmatrix} \right)}.
\end{align*}
\end{proof}

Using the same reasoning as in Propositions  \ref{prop:xtildeR} and \ref{prop:ytildeR} it can be shown that
\begin{align}
    \ddot{\tilde{\bx}}_{ij} &= E \left( \bX^m_i \mid \bx^o_i, \by^o_i, Z_{ij} = 1, \dot{\bvartheta} \right) \nonumber \\ 
    &= \dot{\bmu}^m_{\bX \mid j} + \begin{bmatrix} \dot{\bSigma}^{mo}_{\bX \bX \mid j} & \dot{\bSigma}^{mo}_{\bX \bY \mid j} \end{bmatrix} \begin{bmatrix} \dot{\bSigma}^{oo}_{\bX \bX \mid j} & \dot{\bSigma}^{oo}_{\bX \bY \mid j} \\[1ex] \dot{\bSigma}^{oo}_{\bY \bX \mid j} & \dot{\bSigma}^{oo}_{\bY \bY \mid j} \end{bmatrix}^{-1} \left( \begin{bmatrix} \bx^o_i \\[1ex] \by^o_i \end{bmatrix} - \begin{bmatrix} \dot{\bmu}^o_{\bX \mid j} \\[1ex] \dot{\tilde{\bmu}}^o_{\bY \mid j} \end{bmatrix} \right),  \label{eq:x_tilde} \\[1.5ex]
    \ddot{\tilde{\tilde{\bx}}}_{ij} &= E \left( \bX^m_i \bX^{m \top}_i \mid \bx^o_i, \by^o_i, Z_{ij} = 1, \dot{\bvartheta} \right) \nonumber \\ 
    &= \dot{\bSigma}^{mm}_{\bX \bX \mid j} - \begin{bmatrix} \dot{\bSigma}^{mo}_{\bX \bX \mid j} & \dot{\bSigma}^{mo}_{\bX \bY \mid j} \end{bmatrix} \begin{bmatrix} \dot{\bSigma}^{oo}_{\bX \bX \mid j} & \dot{\bSigma}^{oo}_{\bX \bY \mid j} \\[1ex] \dot{\bSigma}^{oo}_{\bY \bX \mid j} & \dot{\bSigma}^{oo}_{\bY \bY \mid j} \end{bmatrix}^{-1} \begin{bmatrix} \dot{\bSigma}^{om}_{\bX \bX \mid j} \\[1ex] \dot{\bSigma}^{om}_{\bY \bX \mid j} \end{bmatrix} + \ddot{\tilde{\bx}}_{ij} \ddot{\tilde{\bx}}^\top_{ij}, \label{eq:xx_tilde} \\[1.5ex] 
    \ddot{\tilde{\by}}_{ij} &= E \left( \bY^m_i \mid \bx^o_i, \by^o_i, Z_{ij} = 1, \dot{\bvartheta} \right) \nonumber \\
    &= \dot{\tilde{\bmu}}^m_{\bY \mid j} + \begin{bmatrix} \dot{\bSigma}^{mo}_{\bY \bX \mid j} & \dot{\bSigma}^{mo}_{\bY \bY \mid j} \end{bmatrix} \begin{bmatrix} \dot{\bSigma}^{oo}_{\bX \bX \mid j} & \dot{\bSigma}^{oo}_{\bX \bY \mid j} \\[1ex] \dot{\bSigma}^{oo}_{\bY \bX \mid j} & \dot{\bSigma}^{oo}_{\bY \bY \mid j} \end{bmatrix}^{-1} \left( \begin{bmatrix} \bx^o_i \\[1ex] \by^o_i \end{bmatrix} - \begin{bmatrix} \dot{\bmu}^o_{\bX \mid j} \\[1ex] \dot{\tilde{\bmu}}^o_{\bY \mid j} \end{bmatrix} \right), \label{eq:y_tilde} \\[1.5ex]
   \ddot{ \tilde{\tilde{\by}}}_{ij} &= E \left( \bY^m_i \bY^{m \top}_i \mid \bx^o_i, \by^o_i, Z_{ij} = 1, \dot{\bvartheta} \right) \nonumber \\
    &= \dot{\bSigma}^{mm}_{\bY \bY \mid j} - \begin{bmatrix} \dot{\bSigma}^{mo}_{\bY \bX \mid j} & \dot{\bSigma}^{mo}_{\bY \bY \mid j} \end{bmatrix} \begin{bmatrix} \dot{\bSigma}^{oo}_{\bX \bX \mid j} & \dot{\bSigma}^{oo}_{\bX \bY \mid j} \\[1ex] \dot{\bSigma}^{oo}_{\bY \bX \mid j} & \dot{\bSigma}^{oo}_{\bY \bY \mid j} \end{bmatrix}^{-1} \begin{bmatrix} \dot{\bSigma}^{om}_{\bX \bY \mid j} \\[1ex] \dot{\bSigma}^{om}_{\bY \bY \mid j} \end{bmatrix} + \ddot{\tilde{\by}}_{ij} \ddot{\tilde{\by}}^\top_{ij}. \label{eq:yy_tilde}
\end{align}

To prepare for the M-step, the following covariance terms are also computed
\begin{align}
    \ddot{\bOmega}_{\bX \mid ij} &= \text{Cov} \left( \begin{bmatrix} \bx^o_i \\[1ex] \bX^m_i \end{bmatrix} \, \middle| \, \bx^o_i, \by^o_i, Z_{ij} = 1, \dot{\bvartheta} \right) = \begin{bmatrix} \boldsymbol{0}_{d^o_\bX \times d^o_\bX} & \boldsymbol{0}_{d^o_\bX \times d^m_\bX}  \\[1ex]  \boldsymbol{0}_{d^m_\bX \times d^o_\bX} & \ddot{\tilde{\tilde{\bx}}}_{ij} - \ddot{\tilde{\bx}}_{ij} \ddot{\tilde{\bx}}^\top_{ij} \end{bmatrix}, \label{eq:omega_X} \\[1.5ex]
    \ddot{\bOmega}^\ast_{\bX \mid ij} &= \text{Cov} \left( \begin{bmatrix} 1 \\ \bx^o_i \\[1ex] \bX^m_i \end{bmatrix} \, \middle| \, \bx^o_i, \by^o_i, Z_{ij} = 1, \dot{\bvartheta} \right) = \begin{bmatrix} 0 & \boldsymbol{0}_{1 \times d^o_\bX} & \boldsymbol{0}_{1 \times d^m_\bX} \\[1ex] \boldsymbol{0}_{d^o_\bX \times 1} &  \boldsymbol{0}_{d^o_\bX \times d^o_\bX} & \boldsymbol{0}_{d^o_\bX \times d^m_\bX} \\[1ex]  \boldsymbol{0}_{d^m_\bX \times 1} & \boldsymbol{0}_{d^m_\bX \times d^o_\bX} & \ddot{\tilde{\bx}}_{ij} - \ddot{\tilde{\bx}}_{ij} \ddot{\tilde{\bx}}^\top_{ij} \end{bmatrix} = \begin{bmatrix} 0 & \boldsymbol{0}_{1 \times d_\bX} \\[1ex] \boldsymbol{0}_{d_\bX \times 1} &  \ddot{\bOmega}_{\bX \mid ij} \end{bmatrix}, \label{eq:omega_X_star} \\[1.5ex]
    \ddot{\bOmega}_{\bY \mid ij} &= \text{Cov} \left( \begin{bmatrix} \by^o_i \\[1ex] \bY^m_i \end{bmatrix} \, \middle| \, \bx^o_i, \by^o_i, Z_{ij} = 1, \dot{\bvartheta} \right) = \begin{bmatrix} \boldsymbol{0}_{d^o_\bX \times d^o_\bX} & \boldsymbol{0}_{d^o_\bY \times d^m_\bY} \\[2ex]  \boldsymbol{0}_{d^m_\bY \times d^o_\bY} &\ddot{ \tilde{\tilde{\by}}}_{ij} - \ddot{\tilde{\by}}_{ij} \ddot{\tilde{\by}}^\top_{ij} \end{bmatrix}, \label{eq:omega_Y} \\[1.5ex]
    \ddot{\bDelta}_{ij} &= \text{Cov}\left( \begin{bmatrix} \by^o_i \\[1ex] \bY^m_i \end{bmatrix}, \begin{bmatrix} 1 \\ \bx^o_i \\[1ex] \bX^m_i \end{bmatrix} \, \middle| \, \bx^o_i, \by^o_i, Z_{ij} = 1, \dot{\bvartheta} \right) = \begin{bmatrix}
\boldsymbol{0}_{d^o_\bY \times 1} & \boldsymbol{0}_{d^o_\bY \times d^o_\bX} & \boldsymbol{0}_{d^o_\bY \times d^m_\bX}  \\[1ex]
\boldsymbol{0}_{d^m_\bY \times 1} & \boldsymbol{0}_{d^m_\bY \times d^o_\bX} & \ddot{\tilde{\bSigma}}^{mm}_{\bY \bX \mid ij}
\end{bmatrix}, \label{eq:delta_ij} \\[1.5ex]
    \ddot{\tilde{\bSigma}}^{mm}_{\bY \bX \mid ij} &= \begin{bmatrix} \dot{\bSigma}^{mm}_{\bY \bY \mid j} & \dot{\bSigma}^{mm}_{\bY \bX \mid j} \\[1ex] \dot{\bSigma}^{mm}_{\bX \bY \mid j} & \dot{\bSigma}^{mm}_{\bX \bX \mid j} \end{bmatrix} - \begin{bmatrix} \dot{\bSigma}^{mo}_{\bY \bY \mid j} & \dot{\bSigma}^{mo}_{\bY \bX \mid j} \\[1ex]
\dot{\bSigma}^{mo}_{\bX \bY \mid j} & \dot{\bSigma}^{mo}_{\bX \bX \mid j} \end{bmatrix} \begin{bmatrix} \dot{\bSigma}^{oo}_{\bY \bY \mid j} & \dot{\bSigma}^{oo}_{\bY \bX \mid j} \\[1ex] \dot{\bSigma}^{oo}_{\bX \bY \mid j} & \dot{\bSigma}^{oo}_{\bX \bX \mid j} \end{bmatrix}^{-1} \begin{bmatrix} \dot{\bSigma}^{om}_{\bY \bY \mid j} & \dot{\bSigma}^{om}_{\bY \bX \mid j} \\[1ex] \dot{\bSigma}^{om}_{\bX \bY \mid j} & \dot{\bSigma}^{om}_{\bX \bX \mid j} \end{bmatrix}. \label{eq:Sigma_mm_YX}
\end{align}
where $\ddot{\tilde{\bSigma}}^{mm}_{\bY \bX \mid ij} = \text{Cov} \left( \bY^m_i, \bX^m_i \, \middle| \, \bx^o_i, \by^o_i, Z_{ij} = 1, \dot{\bvartheta} \right)$. The value of $\ddot{\tilde{\bSigma}}_{\bY \bX \mid ij}$ can be obtained using the following proposition.

\begin{proposition}
Given observed covariates $\bx^o_i$, observed responses $\by^o_i$, and $Z_{ij} = 1$, the covariance between missing responses $\bY^m_i$ and missing covariates $\bX^m_i$ is given by
\begin{align*}
    \ddot{\tilde{\bSigma}}^{mm}_{\bY \bX \mid ij} &= \text{Cov} \left( \bY^m_i, \bX^m_i \, \middle| \, \bx^o_i, \by^o_i, Z_{ij} = 1, \dot{\bvartheta} \right) \\
    &= \begin{bmatrix} \dot{\bSigma}^{mm}_{\bY \bY \mid j} & \dot{\bSigma}^{mm}_{\bY \bX \mid j} \\[1ex] \dot{\bSigma}^{mm}_{\bX \bY \mid j} & \dot{\bSigma}^{mm}_{\bX \bX \mid j} \end{bmatrix} - \begin{bmatrix} \dot{\bSigma}^{mo}_{\bY \bY \mid j} & \dot{\bSigma}^{mo}_{\bY \bX \mid j} \\[1ex]
\dot{\bSigma}^{mo}_{\bX \bY \mid j} & \dot{\bSigma}^{mo}_{\bX \bX \mid j} \end{bmatrix} \begin{bmatrix} \dot{\bSigma}^{oo}_{\bY \bY \mid j} & \dot{\bSigma}^{oo}_{\bY \bX \mid j} \\[1ex] \dot{\bSigma}^{oo}_{\bX \bY \mid j} & \dot{\bSigma}^{oo}_{\bX \bX \mid j} \end{bmatrix}^{-1} \begin{bmatrix} \dot{\bSigma}^{om}_{\bY \bY \mid j} & \dot{\bSigma}^{om}_{\bY \bX \mid j} \\[1ex] \dot{\bSigma}^{om}_{\bX \bY \mid j} & \dot{\bSigma}^{om}_{\bX \bX \mid j} \end{bmatrix}.
\end{align*}
\end{proposition}
\begin{proof}
The proof is analogous to the one for Proposition \ref{prop:cov}.
\end{proof}

The M-step involves the update of the parameters: $ {\pi_j}, {\bmu}_{\bX \mid j}, {\bSigma}_{\bX \mid j}, {\bbeta}_j$ and ${\bSigma}_{\bY \mid j}$, with  $j=1, \ldots, k$. Specifically,
\begin{align}
    \ddot{\pi}_j &= \dfrac{1}{n}\displaystyle\sum_{i = 1}^n \ddot{\tilde{z}}_{ij}, \label{eq:pi_j} \\[2ex] 
    \ddot{\bmu}_{\bX \mid j} &= \dfrac{\displaystyle\sum_{i = 1}^n \ddot{\tilde{z}}_{ij} \begin{bmatrix} \bx^o_i \\[1ex] \ddot{\tilde{\bx}}_{ij} \end{bmatrix} }{\displaystyle\sum_{i = 1}^n \ddot{\tilde{z}}_{ij}}, \label{eq:mu_X_j} \\[2ex]
    \ddot{\bSigma}_{\bX \mid j} &= \dfrac{\displaystyle\sum_{i = 1}^n \ddot{\tilde{z}}_{ij} \left[ \ddot{\bOmega}_{\bX \mid ij}  + \left( \begin{bmatrix} \bx^o_i \\[1ex] \ddot{\tilde{\bx}}_{ij} \end{bmatrix} - \ddot{\bmu}_{\bX \mid j} \right) \left( \begin{bmatrix} \bx^o_i \\[1ex] \ddot{\tilde{\bx}}_{ij} \end{bmatrix} - \ddot{\bmu}_{\bX \mid j} \right)^\top \right]}{\displaystyle\sum_{i = 1}^n \ddot{\tilde{z}}_{ij}}, \label{eq:Sigma_X_j} \\[2ex]
    \ddot{\bbeta}_j &= \left( \displaystyle\sum_{i = 1}^n \ddot{\tilde{z}}_{ij} \begin{bmatrix} 1 \\ \bx^o_i \\[1ex] \ddot{\tilde{\bx}}_{ij} \end{bmatrix} \begin{bmatrix} 1 & \bx^o_i & \ddot{\tilde{\bx}}_{ij} \end{bmatrix} + \displaystyle\sum_{i = 1}^n \ddot{\tilde{z}}_{ij} \ddot{\bOmega}^\ast_{\bX \mid ij} \right)^{-1} \left( \sum_{i = 1}^n \ddot{\tilde{z}}_{ij} \begin{bmatrix} 1 \\ \bx^o_i \\[1ex] \ddot{\tilde{\bx}}_{ij} \end{bmatrix} \begin{bmatrix} \by^o_i & \ddot{\tilde{\by}}_{ij} \end{bmatrix} + \displaystyle\sum_{i = 1}^n \ddot{\tilde{z}}_{ij} \ddot{\bDelta}^\top_{ij} \right), \label{eq:beta_j} \\[2ex]
    \ddot{\bSigma}_{\bY \mid j} &= \dfrac{\displaystyle\sum_{i = 1}^n \ddot{\tilde{z}}_{ij} \left[ \ddot{\bSigma}^\ast_{\bY \mid ij} + \left( \begin{bmatrix} \by^o_i \\[1ex] \ddot{\tilde{\by}}_{ij} \end{bmatrix} - \ddot{\bbeta}^\top_j \begin{bmatrix} 1 \\ \bx^o_i \\[1ex] \ddot{\tilde{\bx}}_{ij} \end{bmatrix} \right) \left( \begin{bmatrix} \by^o_i \\[1ex] \ddot{\tilde{\by}}_{ij} \end{bmatrix} - \ddot{\bbeta}^\top_j \begin{bmatrix} 1 \\ \bx^o_i \\[1ex] \ddot{\tilde{\bx}}_{ij} \end{bmatrix} \right)^\top \right]}{\displaystyle\sum_{i = 1}^n \ddot{\tilde{z}}_{ij}}, \label{eq:Sigma_Y_j}
\end{align}
where $\ddot{\bSigma}^\ast_{\bY \mid ij} = \left( \ddot{\bOmega}_{\bY \mid ij} + \ddot{\bDelta}_{ij} {\ddot{\bbeta}}_j + \ddot{\bbeta}^\top_j \ddot{\bDelta}^\top_{ij} + \ddot{\bbeta}^\top_j \ddot{\bOmega}^\ast_{\bX \mid ij} \ddot{\bbeta}_j \right)^\top$.
One important note about these results is that all quantities in the E-step and the parameter updates in the M-step are available in closed form, enabling faster computations without any numerical procedures. Parameters can be initialized by applying global mean imputation on the original dataset and obtaining initial group memberships from manual specification or a heuristic clustering procedure. On the now complete dataset, using the obtained membership the $M$-step given in  \citet{dang2017multivariate} is used to obtain the initial parameters.
Algorithm convergence is measured using the Aitken acceleration \citep{aitken_series_1926}; see \citet{mcnicholas_mixture_2020} for further details on the use of the Aitken acceleration in cluster analysis. Algorithm~\ref{alg:EM_MRRC} summarizes the E-step and M-step for fitting the MRRC with missing values to incomplete datasets.

\begin{algorithm}[!t]
\caption{The EM Algorithm for MRRC  with missing values}
\label{alg:EM_MRRC}
\begin{algorithmic}[1]
\State Initialize $ \dot{\bvartheta}$
\Repeat

  \State \textbf{E-step:} For $j = 1, \ldots, k$:
  \State \quad Calculate $\ddot{\tilde{z}}_{ij}$ according to~\eqref{eq:z_tilde} 
  \State \quad Calculate $\ddot{\tilde{\bx}}_{ij}$ and $\ddot{\tilde{\tilde{\bx}}}_{ij}$ according to~\eqref{eq:x_tilde} and ~\eqref{eq:xx_tilde}, respectively
  \State \quad Calculate $\ddot{\tilde{\by}}_{ij}$ and $\ddot{ \tilde{\tilde{\by}}}_{ij}$ according to ~\eqref{eq:y_tilde} and ~\eqref{eq:yy_tilde}, respectively 
  \State \quad Calculate $\ddot{\bOmega}_{\bX \mid ij}$, $\ddot{\bOmega}^\ast_{\bX \mid ij}$, and $\ddot{\bOmega}_{\bY \mid ij}$ according to ~\eqref{eq:omega_X}, ~\eqref{eq:omega_X_star}, and ~\eqref{eq:omega_Y}, respectively
  \State \quad Calculate $\ddot{\bDelta}_{ij}$ and $\ddot{\tilde{\bSigma}}^{mm}_{\bY \bX \mid ij}$ according to~\eqref{eq:delta_ij} and ~\eqref{eq:Sigma_mm_YX}, respectively

  \State \textbf{M-step:} For $j = 1, \ldots, k$:
  \State \quad Update $\ddot{\pi}_j$ according to~\eqref{eq:pi_j}
  \State \quad Update $\ddot{\bmu}_{\bX \mid j}$ and $\ddot{\bSigma}_{\bX \mid j}$ according to~\eqref{eq:mu_X_j} and ~\eqref{eq:Sigma_X_j}, respectively
  \State \quad Update $\ddot{\bbeta}_j$ and $\ddot{\bSigma}_{\bY \mid j}$ according to~\eqref{eq:beta_j} and ~\eqref{eq:Sigma_Y_j}, respectively

\Until{the Aitken's convergence criterion is met or a maximum number of iterations is reached.}

\end{algorithmic}
\end{algorithm}

\subsection{Further aspects}
This section discusses the issues of identifiability and scalability of the proposed model.
The general conditions for the identifiability of a mixture of Gaussian regressions with random covariates have been established by \citet{dang2017multivariate}. When the data contain missing values, a necessary and sufficient condition for the missing-data mechanism to be ignorable in clustering is that the probabilities of missingness are equal across mixture components \citep{sportisse2024model}. By definition, under a MAR mechanism, the probability of missingness depends on the observed data but remains independent of the unobserved data, thereby satisfying the required condition for being ignorable.
Identifiability of the model parameters further requires that (i) the parameters of the marginal mixture distribution are identifiable, and (ii) a total ordering of the mixture densities holds \citep{sportisse2024model}, where the marginal mixture is the model without missing values. Since the marginal mixture is a mixture of Gaussian distributions with random covariates, it is identifiable given conditions in \citet{dang2017multivariate}, which includes condition (ii); thus, the identifiability conditions reduce to those in \citet{dang2017multivariate}. 

The computational scalability of the proposed model is equivalent to that of a Gaussian mixture model. In the Gaussian mixture model, the dominant computational cost arises from updating the component-specific covariance matrices during the EM iterations. In the proposed model, the main computational burden similarly stems from the component-specific covariance matrices—those associated with the covariates and the response—as well as from updating the coefficients $\bbeta_j$. Consequently, the overall computational complexity remains proportional to $k \times d^2$, where $k$ is the number of mixture components and $d = d_{\bX} + d_{\bY}$ denotes the joint dimensionality of the covariates and responses.


\section{Simulation studies}
\label{sec:Simulation_Study}

In this section, we describe the results of two simulation studies conducted in \textsf{R} to investigate the impact of missing values on the MRRC with missing values in clustering and parameter recovery. To evaluate the clustering performance, we compute the adjusted Rand index \citep[ARI;][]{hubert_comparing_1985} of its resulting partition and the true group memberships used when simulating data. 
The ARI is a corrected version of the Rand index \citep{rand_objective_1971} with an expected value of $0$ under random partitions and a value of $1$ under perfect agreement. 
To evaluate the parameter recovery performance, we compute the absolute error for each parameter estimate.
The absolute error, say $\Delta \btheta$, in an estimate $\hat{\btheta}$ of a parameter $\btheta$ with dimension $d_R \times d_C$ is calculated as
\begin{align*}
    \Delta \btheta = \sum_{l = 1}^{d_R} \sum_{s = 1}^{d_C} \lvert \hat{\theta}_{ls} - \theta_{ls} \rvert.
\end{align*}
Here, the MRRC is fitted using two approaches. 
First, we employ the proposed extension that handles MAR values through the EM algorithm, referred to as the MRRC-EM model. 
Second, we apply the MRRC of \citet{dang2017multivariate} to the dataset after global mean imputation, referred to as the MRRC-GMI model.

For both simulation studies, synthetic datasets consisting of $n = 500$ observations are simulated from a two-component mixture ($k = 2$), with component sizes of exactly $n_1 = 150$ and $n_2 = 350$ ($\pi_1 = 0.3$ and $\pi_2 = 0.7$). Two scenarios for the marginal and conditional distributions are considered:
\begin{enumerate}    
\item \textbf{Gaussian:} Both marginal and conditional distributions in each mixture component are Gaussian.    
\item \textbf{Student's $\boldsymbol{t}$:} The marginal distribution, in each mixture component, is Student's $t$ with $7$ degrees of freedom. The conditional distribution in each mixture component is also Student's $t$ but with $5$ degrees of freedom.
\end{enumerate}
Then, for each scenario, $20$ complete datasets are generated. The functions \texttt{rmvtnorm()} and \texttt{rmvt()} from the \prog{R} package \textbf{mvtnorm} \citep{genz_mvtnorm_2021} were used to simulate Gaussian and Student's $t$ realizations, respectively. For each complete dataset above, missing values are then introduced separately to the covariates and responses under the MAR mechanism, using the function \texttt{ampute()} from the \prog{R} package \textbf{mice} \citep{buuren_mice_2011}. Here, we denote the missing rates by $r_\bY$ and $r_\bX$, where $r_\bY$ is the percentage of observations with at least one missing response, and $r_\bX$ is the percentage of observations with at least one missing covariate. Both $r_\bY$ and $r_\bX$ can take values in $\{ 10\%, 30\%, 50\%, 70\% \}$. 
Under each pair $\left( r_\bY, r_\bX \right)$, $10$ incomplete datasets are generated, each with a different arbitrary missing pattern. In total, there are $2 \times 20 \times 4 \times 4 \times 10 = 6400$ incomplete datasets per simulation study.

\subsection{Simulation study 1}

We first consider situations where each observation has both covariates and responses to be bivariate ($d_\bX = d_\bY = 2$). The covariates are generated using the following mean vectors and covariance matrices:
\begin{align*}
    \bmu_{\bX \mid 1} = \begin{bmatrix} 2 \\[1ex] 4 \end{bmatrix} \quad \bmu_{\bX \mid 2} = \begin{bmatrix} 0 \\[1ex] 0 \end{bmatrix}, \quad \bSigma_{\bX \mid 1} = \begin{bmatrix} 2 & 0 \\ 0 & 2 \end{bmatrix}, \quad \text{and}\quad \bSigma_{\bX \mid 2} = \begin{bmatrix} 1 & 0 \\ 0 & 1 \end{bmatrix}.
\end{align*}
The responses are generated using the following regression coefficient and covariance matrices:
\begin{align*}
    \bbeta_1 = \begin{bmatrix} 2 & -2 \\ -0.5 & 1.5 \\ -1 & 2 \end{bmatrix}, \quad \bbeta_2 = \begin{bmatrix} 0 & 1 \\ 2 & 2 \\ -1 & 1.5 \end{bmatrix}, \quad \bSigma_{\bY \mid 1} = \begin{bmatrix} 2 & 1 \\ 1 & 3 \end{bmatrix}, \quad \text{and}\quad \bSigma_{\bY \mid 2} = \begin{bmatrix} 2 & -1 \\ -1 & 3 \end{bmatrix}.
\end{align*}
Figure 1 in the Supplementary Material provides examples of a complete dataset in simulation study 1.

\figurename~\ref{fig:Mean_ARI} summarizes the average ARI under different scenarios and missing rates in the responses ($r_\bX$) and covariates ($r_\bY$). Full results are available in Table 1 in the Supplementary Material. Regarding the MRRC-EM model, for a fixed missing rate in the responses, the ARI decreases linearly as more observations have missing covariates (larger $r_\bX$). Likewise, for a fixed missing rate in the covariates, the ARI is also lower as more observations have missing responses (larger $r_\bY$). 
Under the Student's $t$ scenarios, these trends are maintained, albeit with lower ARI values than those observed in the Gaussian case. 
For the MRRC-GMI model, a similar linearly decreasing trend in the ARI can be observed. The MRRC-GMI model shows better clustering performance in most Gaussian and Student’s t scenarios. 
The MRRC-EM model only performs better in the Student’s $t$ scenarios where the missing rates in the covariates are $r_{\bX} = 50\%$ and $r_{\bX} = 70\%$. 
One explanation for this phenomenon is that global mean imputation can pull incomplete observations closer to the group centroids, making it easier to assign them to the correct groups. 
However, when describing group characteristics effectively, accurate parameter estimation, in addition to clustering performance, should also be emphasized.

The results regarding parameter recovery of both models are given in Tables 3, 4, 5, and 6 in the Supplementary Material. 
Overall, the absolute error in the parameter estimate tends to increase with larger missing rates. When marginal and conditional distributions are not Gaussian, the absolute error also increases. It is important to note that the MRRC-GMI model produces noticeably larger absolute errors than the MRRC-EM model, especially for the covariance matrices and regression coefficients. These results highlight the advantage of the proposed model compared to using imputation as a pre-processing step.

\subsection{Simulation study 2}

We consider higher-dimensional situations where each observation has $d_{\bX} = 6$ covariates and $d_{\bY} = 4$ responses. The mean vectors, covariance matrices, and regression coefficients used to generate the covariates and responses can be found in the Supplementary Material. Figure 2 in the Supplementary Material provides examples of a complete dataset in this simulation study.

figurename~\ref{fig:Mean_ARI10} summarizes the average ARI under different scenarios and missing rates in the responses ($r_\bX$) and covariates ($r_\bY$). Full results are available in Table 2 in the Supplementary Material. The results regarding parameter recovery of both models are given in Tables 7, 8, 9, and 10 in the Supplementary Material. The ARI is generally higher than in Simulation study 1 because more information from the covariates is available. However, the MRRC-EM model achieves much higher ARIs and much lower absolute errors than the MRRC-GMI model, which illustrates the disadvantage of imputation in higher dimensions.


\begin{figure}[!t]
    \centering
    \includegraphics[width=2.55in]{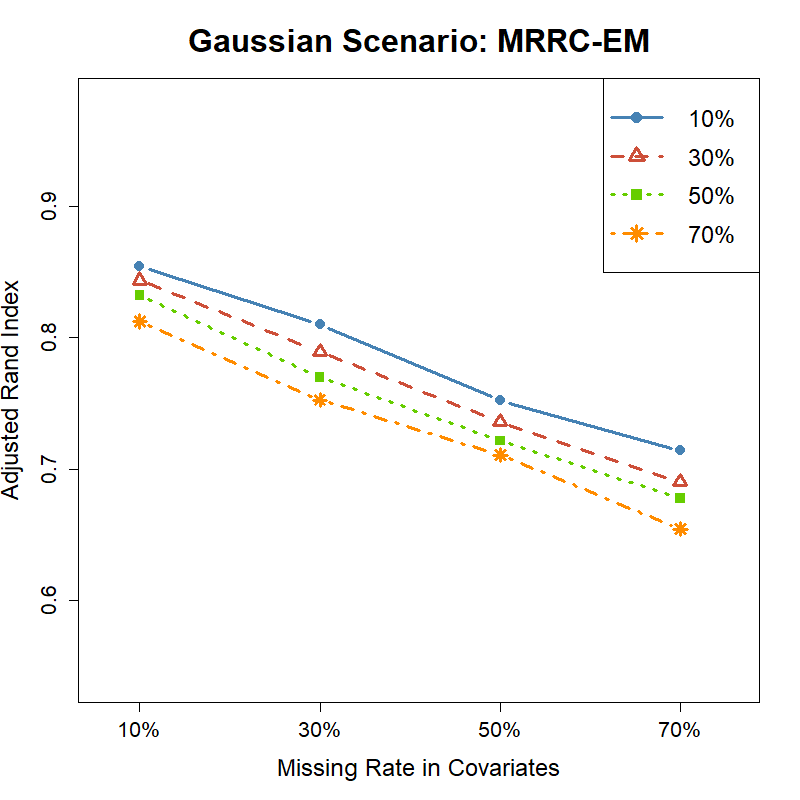}
    \includegraphics[width=2.55in]{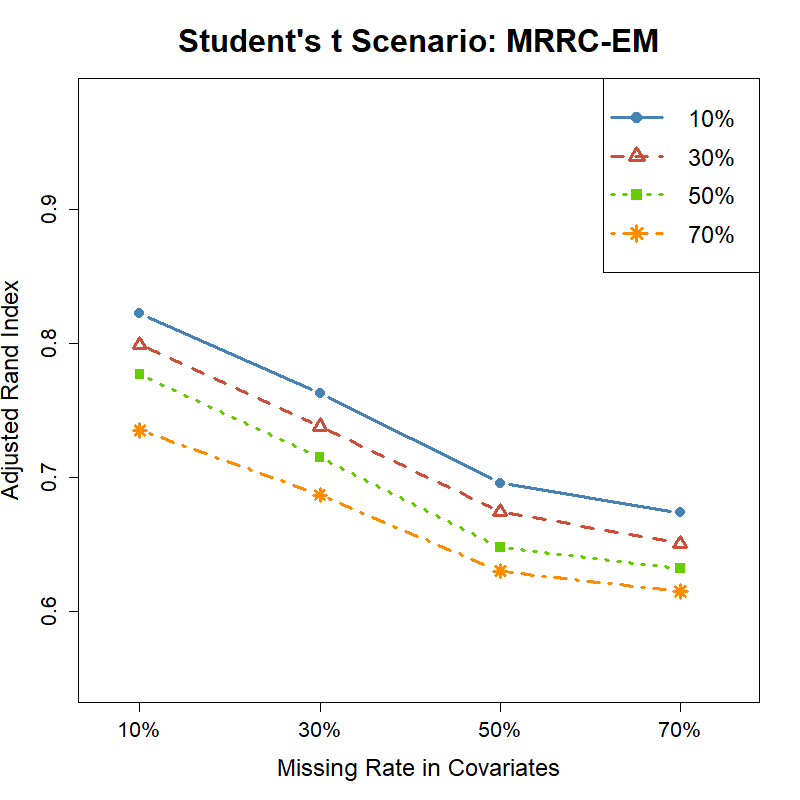}
        \includegraphics[width=2.55in]{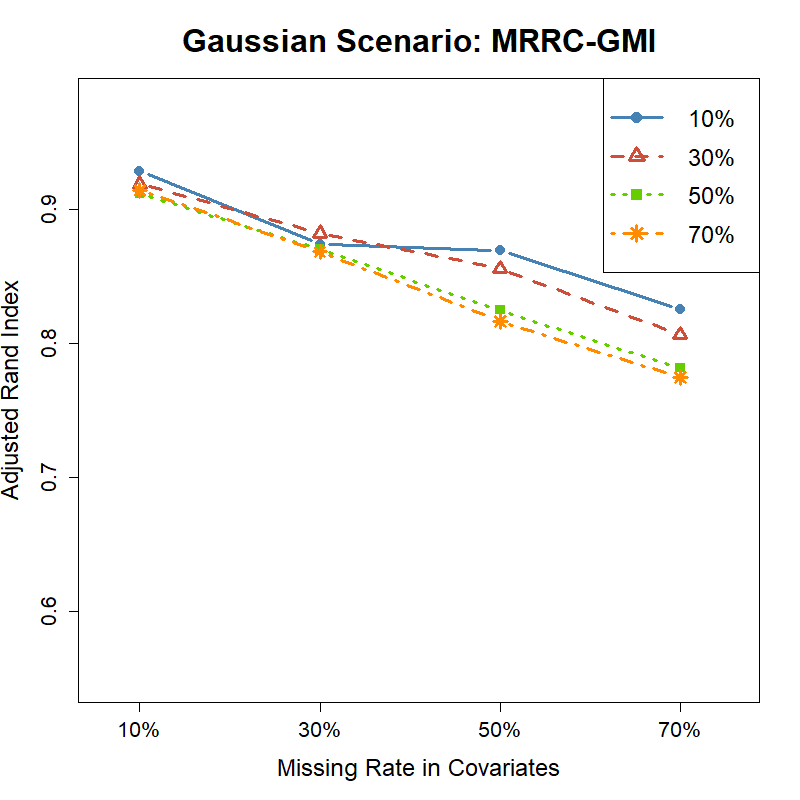}
    \includegraphics[width=2.55in]{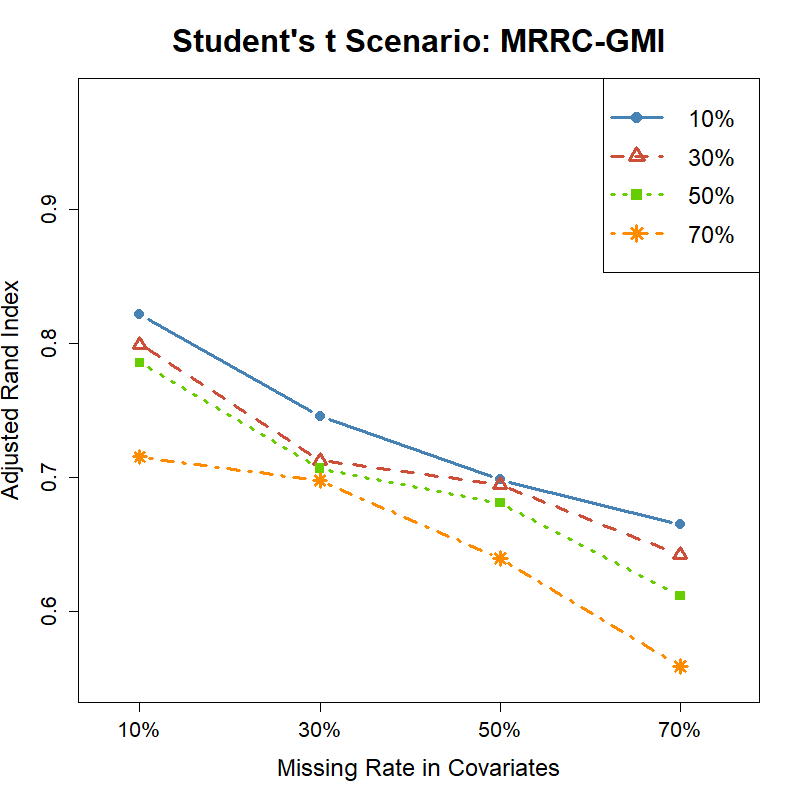}
    \caption{Line plots of the average adjusted Rand index as a function of the covariate missing rate $r_\mathbf{X}$ under different scenarios, based on simulation study 1 ($d_{\mathbf{Y}} = 2$, $d_{\mathbf{X}} = 2$). Colors, line styles, and point shapes represent the response missing rates $r_\mathbf{Y}$.}
    \label{fig:Mean_ARI}
\end{figure}

\begin{figure}[!t]
    \centering
    \includegraphics[width=2.55in]{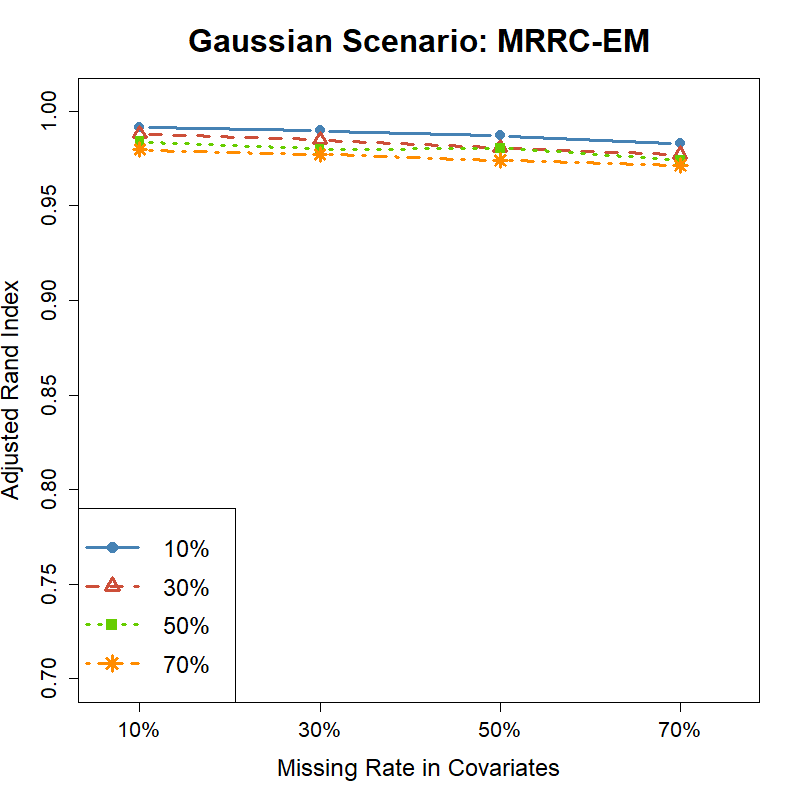}
    \includegraphics[width=2.55in]{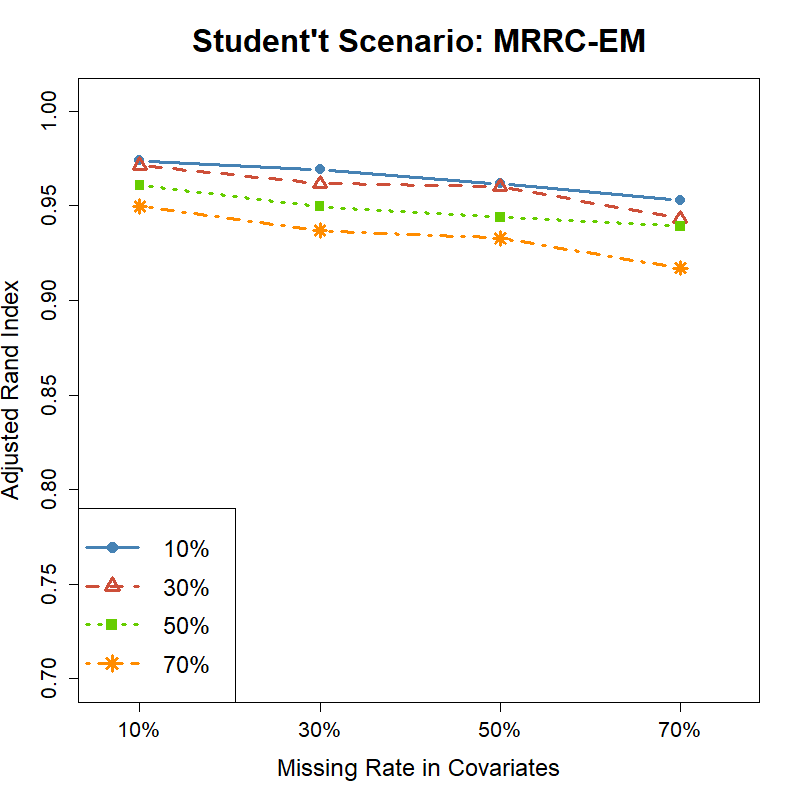}
        \includegraphics[width=2.55in]{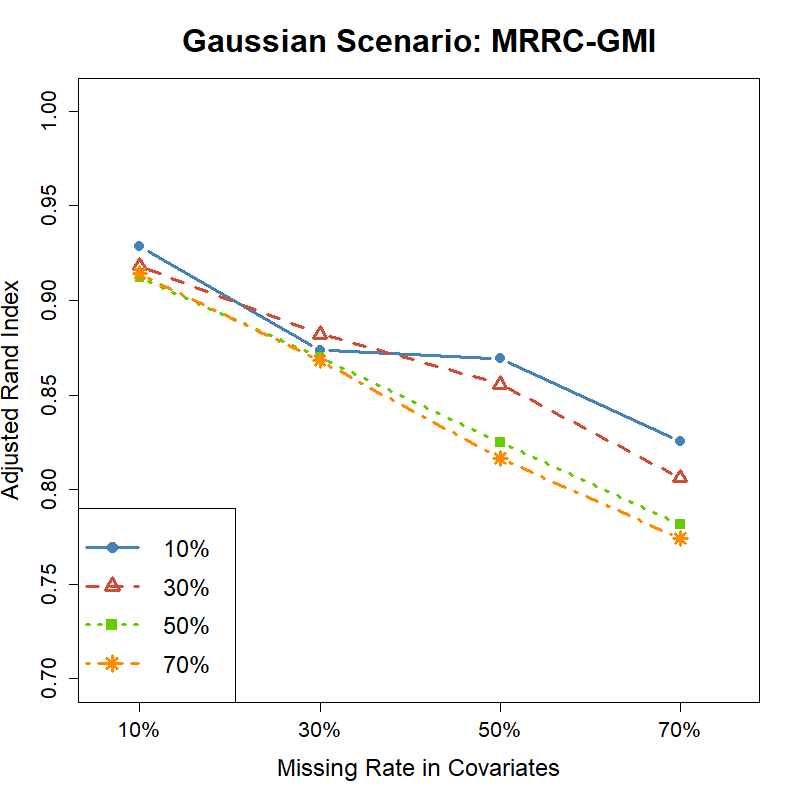}
    \includegraphics[width=2.55in]{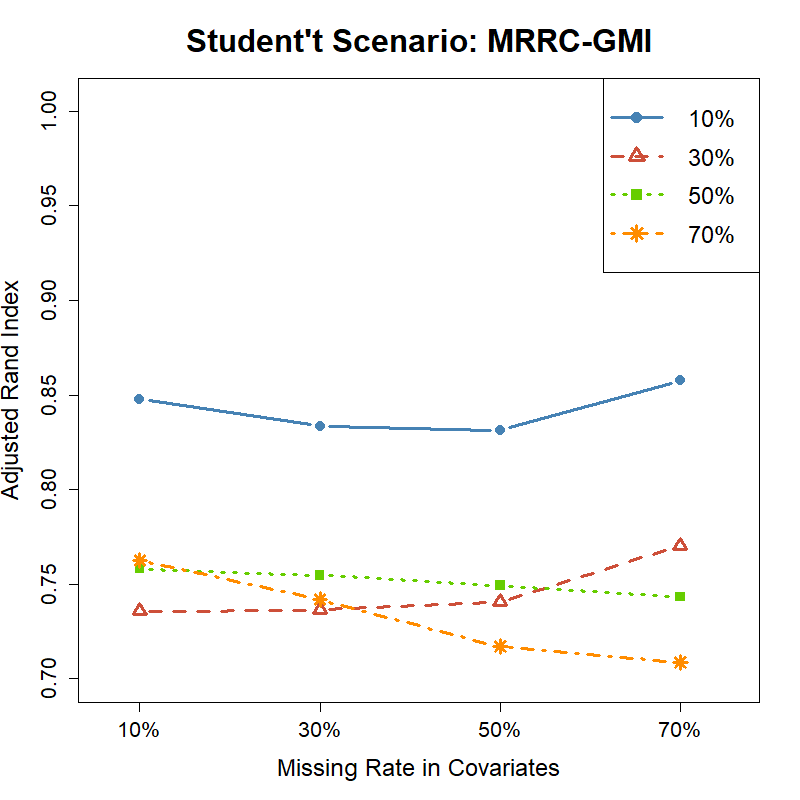}
    \caption{Line plots of the average adjusted Rand index  as a function of the covariate missing rate $r_\mathbf{X}$ under different scenarios, based on simulation study 1 ($d_{\mathbf{Y}} = 4$, $d_{\mathbf{X}} = 6$). Colors, line styles, and point shapes represent the response missing rates $r_\mathbf{Y}$.}
    \label{fig:Mean_ARI10}
\end{figure}


\section{Real data example: Automobile dataset}
\label{sec:real data}

In this section, we demonstrate the proposed methodology using the well-known \textit{Automobile} dataset \citep{schlimmer1985automobile}, which is freely accessible from the UCI Machine Learning Repository at \url{https://archive.ics.uci.edu/dataset/10/automobile}.
The dataset comprises $n = 205$ observations on cars, each described by a set of specifications, including key attributes such as \code{Price}, \code{Symboling rating}, and \code{Normalized losses} (in use when compared to comparable cars). 
\code{Symboling rating} is a categorical variable represented as integer scores ranging from $-3$ to $3$, where higher values indicate greater insurance risk.
We will later use this variable for external validation and interpretation. 
For our analysis, we consider all $d = 15$ continuous variables available in the dataset. 
Among these, \code{Price} ($Y_1$) and \code{Normalized losses} ($Y_2$) are interpreted as response variables, as they are typically seen as outcomes influenced by the intrinsic characteristics of the vehicles. 
The remaining $d_{\bX} = 13$ continuous variables, which describe features such as \code{Engine size}, \code{Horsepower}, dimensions (\code{Length}, \code{Length}, and \code{Width}), and fuel consumption (\code{City MPG} and \code{Highway MPG}), are treated as covariates.
In this context, it is natural to evaluate how these car characteristics influence price and normalized losses through a regression framework, rather than analyzing all variables jointly without distinguishing between predictors and responses. 
The reader should be aware that, in the literature, \code{Price} and \code{Normalized losses} are typically considered separately as response variables to be predicted from the remaining car characteristics, usually within the framework of univariate regression models (see, e.g., \citealp{selvaratnam2021feature} and \citealp{torgo1997functional}).
In contrast, our approach aims to enhance this practice by jointly modeling the two response variables, which allows us to account for potential dependence between them and, crucially, to let such dependence vary across latent clusters that may exist in the data.
Finally, it is also worth noting that both the response variables and several covariates contain missing values, which adds complexity to the analysis and justifies the use of robust methodologies capable of handling incomplete data.
\begin{table}[!ht]
\centering
\caption{Description of the $15$ continuous variables in the Automobile dataset before standardization.}
\label{tab:Automobile_continuous}
\begin{tabular}{llrrrrr}
\toprule
Type                   & Variable                  & No.~of Missing Values & Min     & Max      & Mean & Std.~Dev. \\ 
\midrule
Response  & \code{Normalized losses} & 41     & 65.00   & 256.00   & 122.00      & 35.44            \\
                   & \code{Price}             & 4      & 5118.00 & 45400.00 & 13207.13    & 7947.07          \\ 
                   \midrule
Covariate & \code{Wheel base}         & 0      & 86.60   & 120.90   & 98.76       & 6.02             \\
                   & \code{Length}            & 0      & 141.10  & 208.10   & 174.05      & 12.34            \\
                   & \code{Width}             & 0      & 60.30   & 72.30    & 65.91       & 2.15             \\
                   & \code{Height}           & 0      & 47.80   & 59.80    & 53.72       & 2.44             \\
                   & \code{Curb weight}       & 0      & 1488.00 & 4066.00  & 2555.57     & 520.68           \\
                   & \code{Engine size}       & 0      & 61.00   & 326.00   & 126.91      & 41.64            \\
                   & \code{Bore}              & 4      & 2.54    & 3.94     & 3.33        & 0.27             \\
                   & \code{Stroke}            & 4      & 2.07    & 4.17     & 3.26        & 0.32             \\
                   & \code{Compression ratio} & 0      & 7.00    & 23.00    & 10.14       & 3.97             \\
                   & \code{Horsepower}        & 2      & 48.00   & 288.00   & 104.26      & 39.71            \\
                   & \code{Peak RPM}          & 2      & 4150.00 & 6600.00  & 5125.37     & 479.33           \\
                   & \code{City MPG}          & 0      & 13.00   & 49.00    & 25.22       & 6.54             \\
                   & \code{Highway MPG}       & 0      & 16.00   & 54.00    & 30.75       & 6.89             \\ 
                   \bottomrule
\end{tabular}
\end{table}

In our analysis, each variable is standardized by subtracting its sample mean and dividing by its sample standard deviation. 
This step is essential because, as we can see from \tablename~\ref{tab:Automobile_continuous}, the variables are measured on different scales; standardization ensures that all variables have zero mean and unit variance, making them comparable.
\figurename~\ref{fig:Automobile_scatter1} displays pairwise scatterplots of the standardized data. While both responses are considered, only $6$ out of the $13$ covariates are shown for aesthetic reasons. 
It is also worth noting that, due to the presence of missing values in the dataset (cf.~\tablename~\ref{tab:Automobile_continuous}), fewer than 205 data points are actually plotted in \figurename~\ref{fig:Automobile_scatter1}.
The displayed subset of scatterplots is enough to reveal several key features that motivate our modeling approach. Specifically, a dependence structure is evident between the two responses, and a clear group structure emerges in the data. 
This is particularly evident in the relationship between \code{Price} and several covariates, which supports the inclusion of a regression component in the mixture model. Moreover, associations are also present among the covariates themselves, such as between \code{Horsepower} and \code{City MPG}. These observations justify the use of a random covariate approach, which accounts for assignment dependence in the clustering.
Finally, the presence of missing data highlights the need for appropriate handling methods to preserve as much statistical information as possible.



\begin{figure}[th]
    \centering
    \includegraphics[width=4.30in]{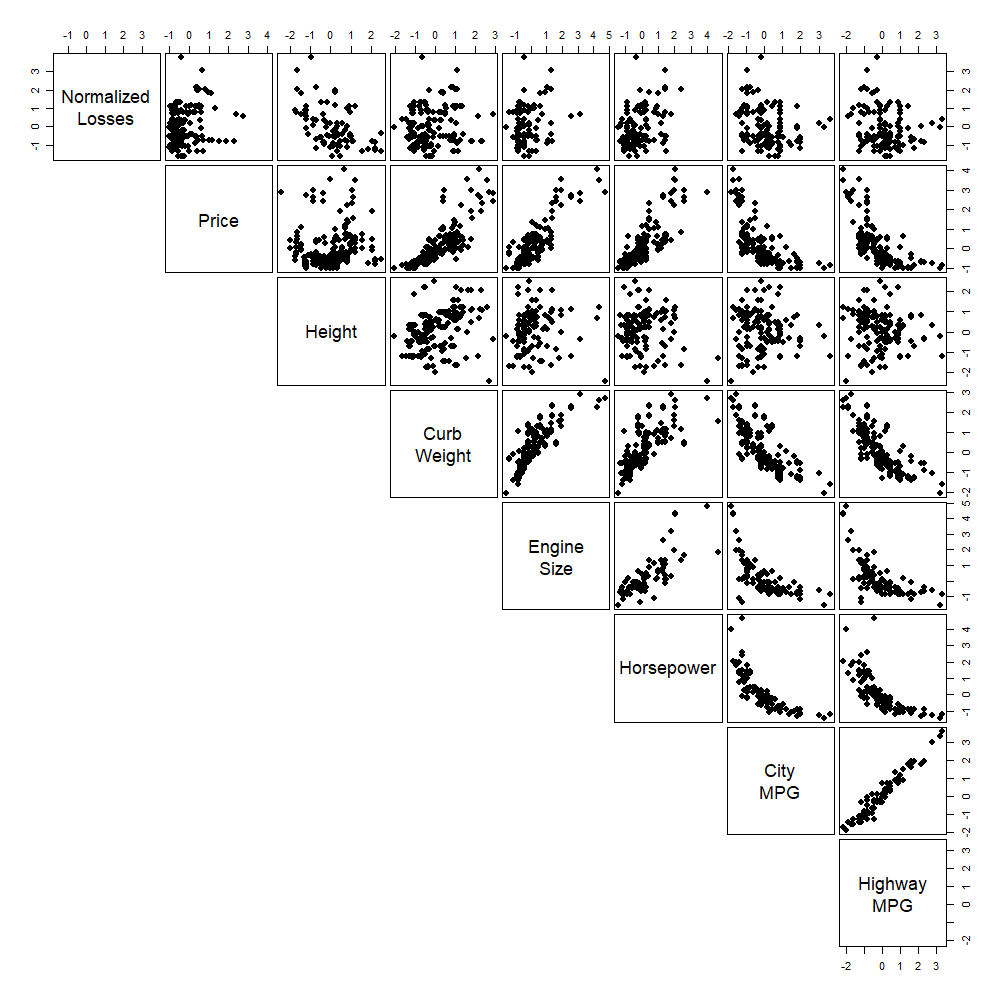}
    \caption{Subset of pairwise scatterplots of the normalized Automobile data, with missing values ignored.}
    \label{fig:Automobile_scatter1}
\end{figure}

We start by estimating the parameters of a Gaussian mixture model (GMM) on all 15 standardized variables, accounting for missing data via ML estimation. 
This is carried out using the EM algorithm proposed by \citet{ghahramani_learning_1994}, as implemented in the \texttt{MGHM()} function from the \textbf{MixtureMissing} package \citep{TongJSS}.
The resulting hard partition is then used to initialize the proposed EM algorithm to find ML estimates of the Gaussian mixture of regressions with random covariates (G-MRRC) model in the presence of MAR values. 
For both the GMM and G-MRRC, the number of components, $k$, is allowed to vary between 1 and 4.

\begin{table}[!ht]
\centering
\caption{Fitting results of the G-MRRC fitted to the normalized Automobile data.}
\label{tab:Automobile_fitting}
\begin{tabular}{lrrr}
\toprule
  $k$      & Number of parameters (\#par) & Log-likelihood & BIC     \\ \midrule
$1$ & $135$               & $-2412.10$       & $5764.55$ \\
$2$ & $271$               & $-1765.49$       & $5119.78$ \\
$3$ & $407$               & $-1429.54$       & $5370.84$ \\
$4$ & $543$               & $-1049.03$       & $5354.86$ \\ \bottomrule
\end{tabular}
\end{table}

In real applications, the number of components $k$ is usually not known \textit{a priori} and thus needs to be chosen using the given data. One common strategy is to fit the assumed mixture model multiple times over a range of values for $k$ and select the one that minimizes the Bayesian information criterion \citep[BIC;][]{schwarz_estimating_1978}. 
Mathematically, $\text{BIC} = -2 l ( \hat{\bvartheta} ) +  \text{\#par} \log n,
$ where $\hat{\bvartheta}$ is the ML estimate of $\bvartheta$ at convergence of the EM algorithm, $l( \hat{\bvartheta} )$ is the associated observed-data log-likelihood, and \#par is the number of parameters in the model, i.e., the cardinality of $\bvartheta$. 
\tablename~\ref{tab:Automobile_fitting} shows the fitting results of the G-MRRC, including the number of parameters, log-likelihood values at convergence, and BIC. 
The optimal solution is $k = 2$, corresponding to the minimum BIC. 
Note that GMM and G-MRRC give the same partition, but G-MRRC offers an advantage in interpretability, as illustrated in the remainder of this section.

There are $81$ and $124$ cars in clusters $1$ and $2$, respectively. 
Tables~\ref{tab:Automobile_make} and \ref{tab:Automobile_symboling} explore \code{Make} and \code{Symboling rating}, which were not used to obtain the clusters, of cars in each cluster.
Cluster $1$ is smaller in size but appears to contain more premium and luxury brands such as Audi, BMW, Jaguar, Mercedes-Benz, Porsche, and Volvo. 
In contrast, cluster $2$ is characterized by more affordable brands, including Chevrolet, Dodge, Honda, Isuzu, Mitsubishi, Nissan, Plymouth, Renault, Subaru, and Toyota. 
In terms of riskiness, cluster $1$ has more low-risk and fewer high-risk cars, which aligns with the presence of luxury brands. 
On the other hand, cluster 2 contains more cars with higher risk ratings, suggesting a potential trade-off between affordability and safety.



\begin{table}[!ht]
\centering
\caption{Makes of cars in each cluster obtained by the G-MRRC with $k = 2$ components of sizes $81$ and $124$, respectively.}
\label{tab:Automobile_make}
\begin{tabular}{rrrlrrr}
  \toprule
\code{Make}          & Cluster 1 & Cluster 2 &  & \code{Make}       & Cluster 1 & Cluster 2 \\ \cline{1-3} \cline{5-7} 
Alfa-Romeo    & 3         & 0         &  & Mitsubishi & 3         & 10        \\
Audi          & 7         & 0         &  & Nissan     & 6         & 12        \\
BMW           & 8         & 0         &  & Peugeot    & 11        & 0         \\
Chevrolet     & 0         & 3         &  & Plymouth   & 1         & 6         \\
Dodge         & 1         & 8         &  & Porsche    & 5         & 0         \\
Honda         & 0         & 13        &  & Renault    & 0         & 2         \\
Isuzu         & 1         & 3         &  & Saab       & 3         & 3         \\
Jaguar        & 3         & 0         &  & Subaru     & 0         & 12        \\
Mazda         & 2         & 15        &  & Toyota     & 5         & 27        \\
Mercedes-Benz & 8         & 0         &  & Volkswagen & 2         & 10        \\
Mercury       & 1         & 0         &  & Volvo      & 11        & 0         \\ \bottomrule
\end{tabular}
\end{table}


\begin{table}[!ht]
\centering
\caption{Symboling ratings of cars in each cluster obtained by the G-MRRC.
A symboling rating ranges from $-3$ to $3$ in which the higher end implies a riskier auto.}
\label{tab:Automobile_symboling}
\begin{tabular}{rrr}
  \toprule
\code{Symboling rating} & Cluster 1 & Cluster 2 \\ 
  \midrule
  $-3$ & $0$ & $0$ \\
$-2$ &   $3$ &   $0$ \\ 
  $-1$ &  $14$ &   $8$ \\ 
  $0$ &  $30$ &  $37$ \\ 
  $1$ &   $9$ &  $45$ \\ 
  $2$ &   $6$ &  $26$ \\ 
  $3$ &  $19$ &   $8$ \\ 
   \midrule
   Total & $81$ & $124$ \\ \bottomrule
\end{tabular}
\end{table}


\figurename~\ref{fig:Autombile_scatter MRRC} presents the ``model-based'' counterpart of \figurename~\ref{fig:Automobile_scatter1}, with points colored according to the maximum \textit{a posteriori} partition obtained from the best-fitting two-component G-MRRC model. In these pairwise scatterplots, missing values are now imputed based on the model estimates from the E-step of the EM algorithm, and the number of observations displayed thus matches the full sample size ($n = 205$). 
A clear separation between the two groups is evident across all displayed variables. Among the response variables, \code{Price} demonstrates greater discriminative power than \code{Normalized~Losses}. 
Furthermore, all covariates---both individually and in pairwise combinations---contribute meaningfully to the resulting partition, providing strong evidence of assignment dependence in the clustering process.
\begin{figure}[!ht]
    \centering
    \includegraphics[width=4.30in]{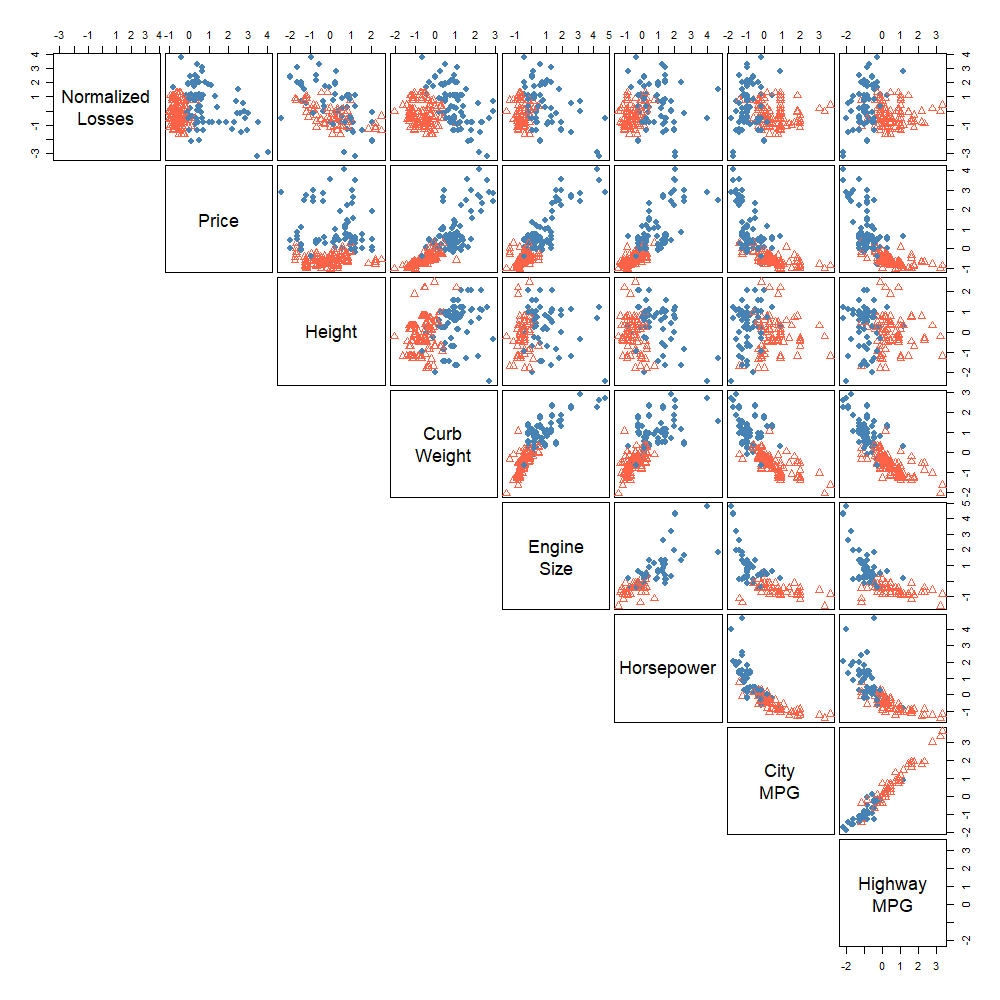}
    \caption{Subset of pairwise scatterplots of the normalized Automobile data, with missing values imputed and points colored according to the best-fitting two-component G-MRRC model.}
    \label{fig:Autombile_scatter MRRC}
\end{figure}


\figurename~\ref{fig:auto_line_plot} visualizes the mean values of various standardized covariates for each cluster. 
The plot reveals noticeable differences between the two clusters across most covariates, except for \code{Height}, \code{Stroke}, \code{Compression ratio}, and \code{Peak RPM}. 
The similarity in these covariates may reflect common design standards and engineering constraints rather than factors like cost or riskiness. 
As expected, cluster $1$, which includes luxury cars, consists of more powerful vehicles with larger engines. 
Cluster $2$ contains smaller cars with better fuel consumption.
\begin{figure}[!ht]
    \centering
    \includegraphics[width=3.5in]{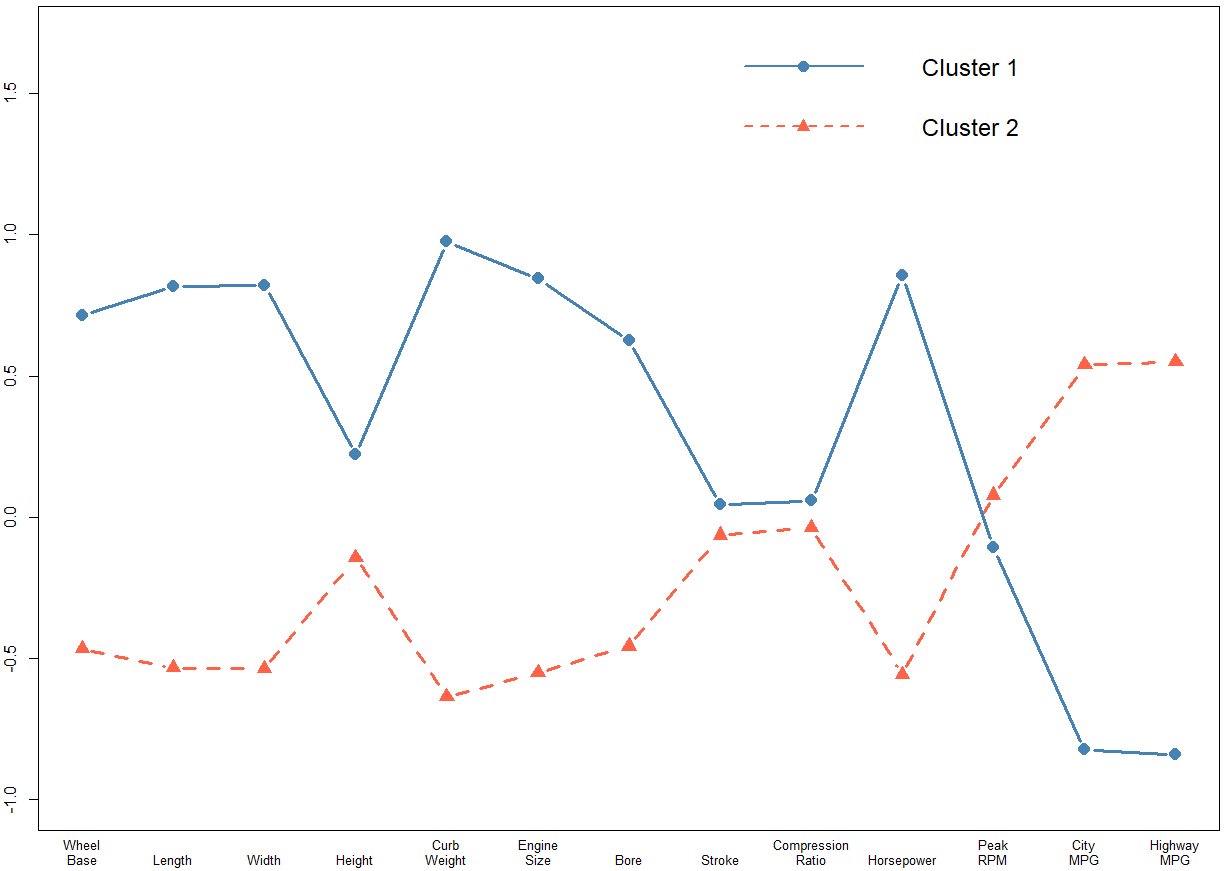}
    \caption{Line plot showing mean values of covariates according to the G-MRRC with $k = 2$.}
    \label{fig:auto_line_plot}
\end{figure}

\tablename~\ref{tab:Automobile_beta} presents coefficient estimates of standardized covariates for each cluster. 
Each coefficient estimate describes the average change in the response for a one-unit change in the covariate while holding other covariates constant. 
In cluster $1$, higher \code{Curb weight} and \code{Height} are associated with higher \code{Normalized losses}, while \code{Engine size} decreases \code{Normalized losses} and increases \code{Price}. 
In cluster $1$, among luxury cars, \code{City MPG} lowers both \code{Normalized losses} and \code{Price}. 
In cluster $2$, engine size contributes to higher \code{Normalized losses} and lower \code{Price}; \code{Highway MPG} is also associated with higher \code{Normalized losses}, while \code{Horsepower} is linked to higher \code{Price}. 
\code{Curb weight} slightly increases the \code{Price} in this cluster as well.
\begin{table}[!ht]
\centering
\caption{Coefficient of standardized covariates in each cluster obtained by the G-MRRC with $k = 2$.}
\label{tab:Automobile_beta}
\begin{tabular}{rrrlrr}
\toprule
\multicolumn{1}{l}{} & \multicolumn{2}{r}{Cluster $1$} & \multicolumn{1}{r}{} & \multicolumn{2}{r}{Cluster $2$} \\ 
\cline{2-3} \cline{5-6} 
                     & \code{Normalized Losses}   & \code{Price}   &                      & \code{Normalized Losses}   & \code{Price}   \\ 
                     \midrule 
Intercept            & $1.90$                & $-0.27$   &                      & $-0.13$               & $-0.42$   \\
\code{Wheel base}           & $0.42$                & $-0.09$   &                      & $-0.82$               & $-0.02$   \\
\code{Length}               & $-0.93$               & $0.17$    &                      & $0.53$                & $0.06$    \\
\code{Width}                & $0.18$                & $0.14$    &                      & $-0.18$               & $0.10$    \\
\code{Height}               & $-0.81$               &$ 0.11$    &                      & $-0.37$               & $0.01$    \\
\code{Curb weight}          & $0.66$                & $-0.30$   &                      & $-0.06$               & $0.26$    \\
\code{Engine size}          & $-0.51$               & $0.80$    &                      & $0.99$                & $-0.49$   \\
\code{Bore}                 & $-0.57$               & $0.12$    &                      & $-0.38$               & $0.08$    \\
\code{Stroke}               & $0.08$                & $-0.06$   &                      & $-0.15$               & $0.08$    \\
\code{Compression ratio}    & $-0.27$               & $0.38$    &                      & $0.03$                & $0.04$    \\
\code{Horsepower}           & $-0.27$               & $-0.12$   &                      & $-0.02$               & $0.28$    \\
\code{Peak RPM}             & $-0.11$               & $0.12$    &                      & $0.18$                & $-0.02$   \\
\code{City MPG}             & $-1.15$               & $-1.48$   &                      & $-0.28$               & $0.03$    \\
\code{Highway MPG}          & $1.80$                & $0.73$    &                      & $0.21$                & $-0.03$   \\ \bottomrule
\end{tabular}
\end{table}

\tablename~\ref{tab:Automobile_SigmaY} shows the estimated correlation matrix of the response variables within each cluster. 
In both clusters, there is a mild negative correlation between \code{Normalized losses} and \code{Price}. 
\figurename~\ref{fig:Automobile_SigmaX} shows the correlation matrix of the covariates for each cluster, indicating that the linear relationships among most covariates are consistent across clusters. 
Specifically, \code{Horsepower} is more positively correlated with other covariates in Cluster $2$.
\begin{table}[!ht]
\centering
\caption{Correlation matrix of the responses in each cluster obtained by the G-MRRC with $k = 2$.}
\label{tab:Automobile_SigmaY}
\begin{tabular}{rrrlrr}
\toprule
\multicolumn{1}{l}{} & \multicolumn{2}{r}{Cluster $1$} &  & \multicolumn{2}{r}{Cluster $2$} \\ \cline{2-3} \cline{5-6} 
                     & \code{Normalized Losses}  & \code{Price}    &  & \code{Normalized Losses}  & \code{Price}    \\ \midrule
\code{Normalized Losses}    & $1.000$            & $-0.055$ &  & $1.000$            & $-0.092$ \\
\code{Price}                & $-0.055$           & $1.000$  &  & $-0.092$           & $1.000$  \\ \bottomrule
\end{tabular}
\end{table}


\begin{figure}[!t]
    \centering
    \includegraphics[width=3.2in]{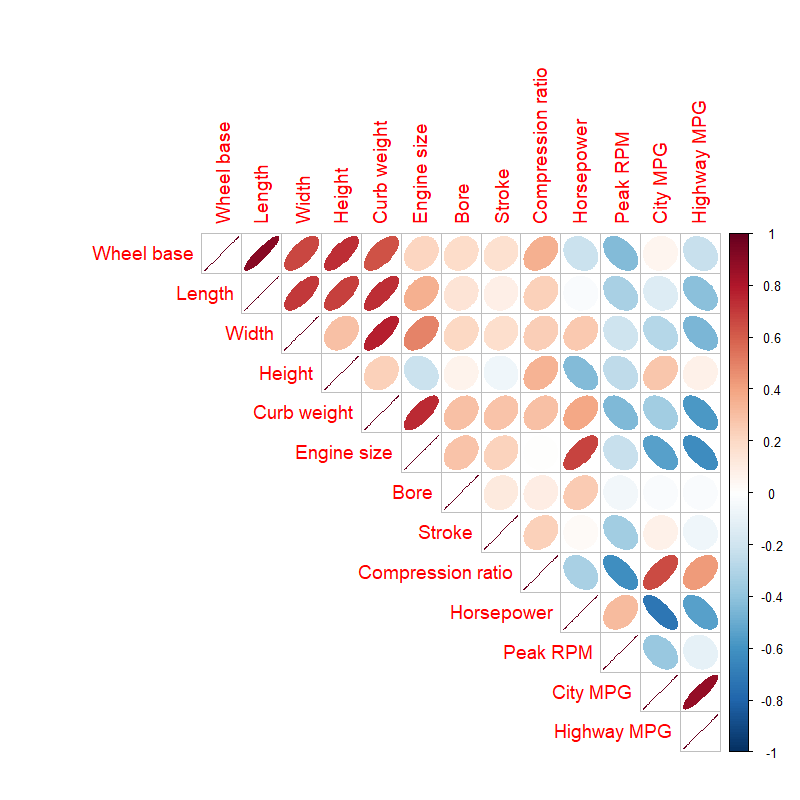}
    \includegraphics[width=3.2in]{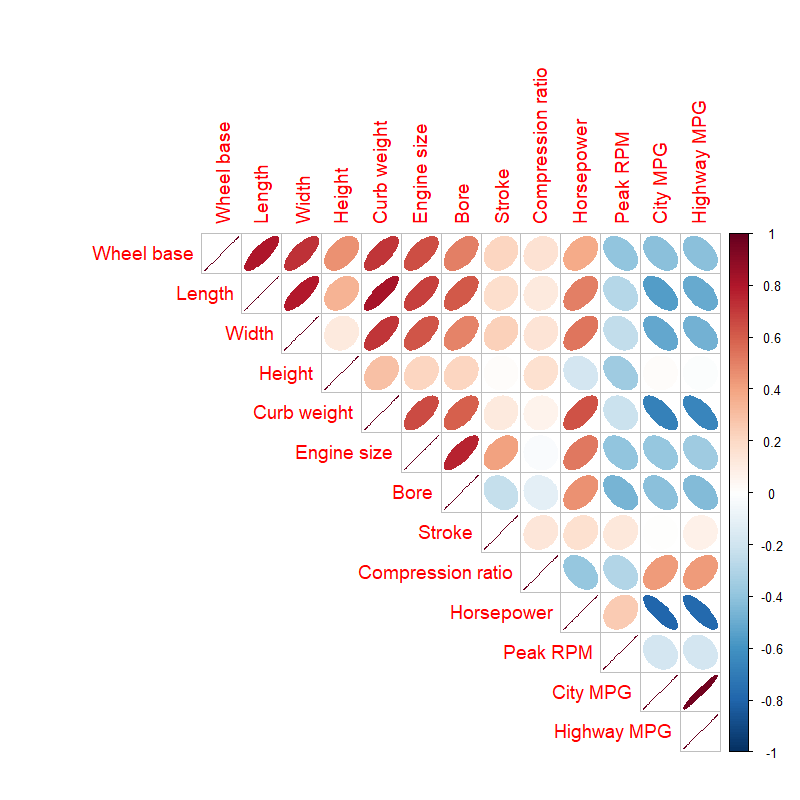}
    \caption{Correlation plots of the covariates in each cluster obtained by the G-MRRC with $k = 2$.}
    \label{fig:Automobile_SigmaX}
\end{figure}




\section{Conclusion and future developments}
\label{sec:conclusion}

This paper addressed the challenge of estimating multivariate linear regression models with multiple random covariates in the presence of missing-at-random (MAR) values affecting both response and covariate spaces. We proposed a maximum likelihood (ML) framework grounded in the conditional-marginal factorization of a multivariate Gaussian distribution, which enhances interpretability by preserving the distinct roles of responses and covariates. The expectation-maximization (EM) algorithm was employed for parameter estimation and missing value imputation, offering a principled solution that makes full use of the available data.

We further extended this framework to model-based clustering by introducing a mixture of multivariate linear regressions with multiple random covariates. 
This formulation accommodates MAR values in both responses and covariates, enabling soft clustering under incomplete data and overcoming limitations of fixed-covariate mixture models. By allowing covariates to contribute to the clustering process—through assignment dependence—the method supports more realistic and informative partitions of heterogeneous data.

Empirical evaluations through simulation studies confirmed the robustness and accuracy of the proposed approach in terms of clustering quality and parameter recovery, even under moderate departures from the Gaussian distribution. The real-data application to the Automobile dataset further demonstrated the method’s effectiveness in practical scenarios with substantial missingness.

Future developments of this work will aim to extend the current framework in several directions. A particularly promising avenue involves relaxing the Gaussian assumption through the adoption of copula-based formulations or more flexible distributions, such as the contaminated normal, the  Student’s $t$, or other heavy-tailed and asymmetric distributions, to enhance robustness against outliers and non-Gaussian behaviors. The Gaussian assumption could be relaxed for the $\bX$, the $\bY$, or both.
Nevertheless, it is important to recognize that one of the main strengths of the proposed methodology lies in the availability of closed-form expressions for all E-step quantities and M-step updates. This property results from the analytical tractability of the Gaussian distribution, under which the joint distribution of $\bX$ and $\bY$ remains Gaussian. When the Gaussian assumption is relaxed, this convenient structure no longer holds, making the extension of the current model to non-Gaussian distributions a challenging task. A natural first step in this direction is to explore distributions that can be formulated as discrete-scale mixtures of normal distributions. Additionally, the incorporation of regularization techniques could facilitate high-dimensional settings where the number of variables is large relative to the sample size. Bayesian formulations may also be explored to provide posterior uncertainty quantification and integrate prior knowledge in both regression and clustering settings. Finally, algorithmic improvements—such as accelerated EM variants or variational approximations—may enhance computational efficiency, especially for large-scale or streaming datasets. These extensions would broaden the applicability of the proposed methodology to a wider range of complex, real-world problems involving incomplete and heterogeneous data.


\section*{Acknowledgments}

This work was supported by \textit{NSF grant N0. 2209974} (Dr. Cristina Tortora). Dr. Antonio Punzo was funded by the European Union - NextGenerationEU, Mission 4, Component 2, in the framework of the GRINS - Growing Resilient, INclusive and Sustainable project (GRINS PE00000018 – CUP E63C22002120006). The views and opinions expressed are solely those of the authors and do not necessarily reflect those of the European Union, nor can the European Union be held responsible for them.


\bibliography{bib}

\end{document}